\documentclass[pra,twocolumn,preprintnumbers,amsmath,amssymb,superscriptaddress,footinbib]{revtex4-1}
\usepackage{amsmath}
\usepackage{enumerate}
\usepackage{amsfonts}
\usepackage{amsthm}
\usepackage{color}
\usepackage{graphicx}
\usepackage{dcolumn}
\usepackage{bm}
\usepackage{bbm}
\usepackage{times}
\usepackage{changes}
\usepackage{comment}
\definecolor{myurlcolor}{rgb}{0,0,0.7}
\definecolor{myrefcolor}{rgb}{0.8,0,0}
\usepackage{hyperref}
\hypersetup{colorlinks, linkcolor=myrefcolor,
citecolor=myurlcolor, urlcolor=myurlcolor,pdfstartview=FitH}
\usepackage[draft]{fixme}
\usepackage{xr}
\usepackage{xspace}

\def\appendix{Appendix\xspace}

\begin{document}

\theoremstyle{plain}
\newtheorem{theorem}{Theorem}
\newtheorem{lemma}[theorem]{Lemma}
\newtheorem{lemm}{Lemma}
\numberwithin{lemm}{section}
\newtheorem{corollary}[theorem]{Corollary}
\newtheorem{conjecture}[theorem]{Conjecture}
\newtheorem{proposition}[theorem]{Proposition}
\newtheorem{corol}[lemm]{Corollary}
\newcommand{\bs}{\boldsymbol}
 \newcommand{\C}{\mathbb{C}}
  \newcommand{\F}{\mathbb{F}}
  \newcommand{\N}{\mathbb{N}}
  \renewcommand{\P}{\mathbb{P}}
  \newcommand{\R}{\mathbb{R}}
  \newcommand{\Z}{\mathbf{Z}}
  \renewcommand{\a}{\mathbf{a}}
  \renewcommand{\b}{\mathbf{b}}
  \renewcommand{\i}{\mathbf{i}}
  \renewcommand{\j}{\mathbf{j}}
  \renewcommand{\c}{\mathbf{c}}
  \newcommand{\e}{\mathbf{e}}
  \newcommand{\f}{\mathbf{f}}
  \newcommand{\g}{\mathbf{g}}
  \newcommand{\gl}{\mathbf{GL}}
  \newcommand{\m}{\mathbf{m}}
  \newcommand{\n}{\mathbf{n}}
  \newcommand{\bNP}{\mathbf{NP}}
  \newcommand{\bNPC}{\mathbf{NPC}}
  \newcommand{\p}{\mathbf{p}}
  \newcommand{\bP}{\mathbb{P}}
  \newcommand{\bPo}{\mathbf{Po}}
  \newcommand{\q}{\mathbf{q}}
  \newcommand{\s}{\mathbf{s}}
  \newcommand{\bt}{\mathbf{t}}
  \newcommand{\T}{\mathbf{T}}
  \newcommand{\U}{\mathbf{U}}
  \renewcommand{\u}{\mathbf{u}}
  \renewcommand{\v}{\mathbf{v}}
  \newcommand{\V}{\mathbf{V}}
  \newcommand{\w}{\mathbf{w}}
  \newcommand{\W}{\mathbf{W}}
  \newcommand{\x}{\mathbf{x}}
  \newcommand{\X}{\mathbf{X}}
  \newcommand{\y}{\mathbf{y}}
  \newcommand{\Y}{\mathbf{Y}}
  \newcommand{\z}{\mathbf{z}}
  \newcommand{\0}{\mathbf{0}}
  \newcommand{\1}{\mathbf{1}}
  \newcommand{\Gam}{\mathbf{\Gamma}}
  \newcommand{\bGamma}{\Gam}
  \newcommand{\Lam}{\mathbf{\Lambda}}
  \newcommand{\lam}{\mbox{\boldmath{$\lambda$}}}
    \newcommand{\om}{\mbox{\boldmath{$\omega$}}}
  \newcommand{\bA}{\mathbf{A}}
  \newcommand{\bB}{\mathbf{B}}
  \newcommand{\bC}{\mathbf{C}}
  \newcommand{\bH}{\mathbf{H}}
  \newcommand{\bL}{\mathbf{L}}
  \newcommand{\bM}{\mathbf{M}}
  \newcommand{\bc}{\mathbf{c}}
  \newcommand{\cA}{\mathcal{A}}
  \newcommand{\cB}{\mathcal{B}}
  \newcommand{\cC}{\mathcal{C}}
  \newcommand{\cD}{\mathcal{D}}
  \newcommand{\cE}{\mathcal{E}}
  \newcommand{\cF}{\mathcal{F}}
  \newcommand{\cG}{\mathcal{G}}
  \newcommand{\cH}{\mathcal{H}}
  \newcommand{\cI}{\mathcal{I}}
  \newcommand{\cJ}{\mathcal{J}}
  \newcommand{\cL}{\mathcal{L}}
  \newcommand{\cM}{\mathcal{E}}
  \newcommand{\cMm}{\mathcal{M}}
    \newcommand{\cN}{\mathcal{N}}
  \newcommand{\cO}{\mathcal{O}}
  \newcommand{\cP}{\mathcal{P}}
  \newcommand{\cQ}{\mathcal{Q}}
  \newcommand{\cR}{\mathcal{R}}
  \newcommand{\br}{\mathbf{r}}
  \newcommand{\cS}{\mathcal{S}}
  \newcommand{\cT}{\mathcal{T}}
  \newcommand{\cU}{\mathcal{U}}
  \newcommand{\cV}{\mathcal{V}}
  \newcommand{\cW}{\mathcal{W}}
  \newcommand{\cX}{\mathcal{X}}
  \newcommand{\cY}{\mathcal{Y}}
  \newcommand{\cZ}{\mathcal{Z}}
  \newcommand{\rE}{\mathrm{E}}
  \newcommand{\rH}{\mathrm{H}}
  \newcommand{\rU}{\mathrm{U}}
  \newcommand{\mfU}{\mathfrak{U}}
  \newcommand{\Cp}{\mathrm{Cap\;}}
  \newcommand{\lan}{\langle}
  \newcommand{\ran}{\rangle}
  \newcommand{\an}[1]{\lan#1\ran}
  \def\diag{\mathop{{\rm diag}}\nolimits}
  \newcommand{\hs}{\hspace*{\parindent}}
  \newcommand{\cl}{\mathop{\mathrm{Cl}}\nolimits}
  \newcommand{\tr}{\mathop{\mathrm{Tr}}\nolimits}
  \newcommand{\Aut}{\mathop{\mathrm{Aut}}\nolimits}
  \newcommand{\argmax}{\mathop{\mathrm{arg\,max}}}
  \newcommand{\Eig}{\mathop{\mathrm{Eig}}\nolimits}
  \newcommand{\Gr}{\mathop{\mathrm{Gr}}\nolimits}
  \newcommand{\Fr}{\mathop{\mathrm{Fr}}\nolimits}
  \newcommand{\trans}{^\top}
  \newcommand{\opt}{\mathop{\mathrm{opt}}\nolimits}
  \newcommand{\per}{\mathop{\mathrm{perm}}\nolimits}
  \newcommand{\haff}{\mathrm{haf\;}}
  \newcommand{\perio}{\mathrm{per}}
  \newcommand{\conv}{\mathrm{conv\;}}
  \newcommand{\Cov}{\mathrm{Cov}}
  \newcommand{\inter}{\mathrm{int}}
  \newcommand{\dist}{\mathrm{dist}}
  \newcommand{\inn}{\mathrm{in}}
  \newcommand{\grank}{\mathrm{grank}}
  \newcommand{\mrank}{\mathrm{mrank}}
  \newcommand{\krank}{\mathrm{krank}}
  \newcommand{\out}{\mathrm{out}}
  \newcommand{\orient}{\mathrm{orient}}
  \newcommand{\Pu}{\mathrm{Pu}}
  \newcommand{\rdc}{\mathrm{rdc}}
  \newcommand{\range}{\mathrm{range\;}}
  \newcommand{\Sing}{\mathrm{Sing\;}}
  \newcommand{\topo}{\mathrm{top}}
  \newcommand{\undir}{\mathrm{undir}}
  \newcommand{\Var}{\mathrm{Var}}
  \newcommand{\rC}{\mathrm{C}}
  \newcommand{\rF}{\mathrm{F}}
  \newcommand{\rL}{\mathrm{L}}
  \newcommand{\rM}{\mathrm{M}}
  \newcommand{\rO}{\mathrm{O}}
  \newcommand{\rR}{\mathrm{R}}
  \newcommand{\rS}{\mathrm{S}}
  \newcommand{\rT}{\mathrm{T}}
  \newcommand{\pr}{\mathrm{pr}}
  \newcommand{\inte}{\mathrm{int}}
  \newcommand{\inv}{\mathrm{inv}}
  \newcommand{\pers}{\per_s}
  \newcommand{\del}{\boldsymbol{\delta}}
  \renewcommand{\alph}{\boldsymbol{\alpha}}
  \newcommand{\bet}{\boldsymbol{\beta}}
  \newcommand{\gam}{\boldsymbol{\gamma}}
  \newcommand{\sig}{\boldsymbol{\sigma}}
  \newcommand{\zet}{\boldsymbol{\zeta}}
  \newcommand{\et}{\boldsymbol{\eta}}
  \newcommand{\xit}{\boldsymbol{\xi}}
  \newcommand{\perm}{\mathrm{perm\;}}
  \newcommand{\adj}{\mathrm{adj\;}}
  \newcommand{\rank}{\mathrm{rank\;}}
  \newcommand{\set}[1]{\{#1\}}
  \newcommand{\spec}{\mathrm{spec\;}}
  \newcommand{\supp}{\mathrm{supp\;}}
  \newcommand{\Tr}{\mathrm{Tr\;}}
  \newcommand{\vol}{\text{vol}}
  \newcommand{\be}{\begin{equation}}
  \newcommand{\ee}{\end{equation}}
	\newcommand{\ben}{\begin{eqnarray}}
  \newcommand{\een}{\end{eqnarray}}
\newcommand{\<}{\langle}
\renewcommand{\>}{\rangle}
\newcommand*{\ket}[1]{\left|#1\right\rangle}
\newcommand*{\bra}[1]{\left\langle #1\right|}
\newcommand*{\proj}[1]{\ket{#1}\bra{#1}}
\newcommand*{\eins}{\ensuremath{\mathbbm 1}}
\def\ot{\otimes}
\newcommand{\mh}[1]{\textcolor{blue}{\tt michal: #1}}

\newenvironment{rcases}
  {\left.\begin{aligned}}
  {\end{aligned}\right\rbrace}

	\newcommand{\nd}{\textendash}
  \newcommand{\md}{\textemdash}

\theoremstyle{definition}
\newtheorem{definition}{Definition}

\theoremstyle{remark}
\newtheorem*{remark}{Remark}
\newtheorem{example}{Example}

\title{The Conditional Uncertainty Principle}

\author{Gilad Gour}
\affiliation{Department of Mathematics and Statistics and Institute for Quantum Science and Technology, University of Calgary, 2500 University Drive NW, Calgary, Alberta, Canada T2N 1N4}

\author{Andrzej Grudka}
\affiliation{Faculty of Physics, Adam Mickiewicz University, 61-614 Pozna\'{n}, Poland}

\author{Micha\l\ Horodecki}
\affiliation{Institute for Theoretical Physics and Astrophysics,
University of Gda{\'n}sk, 80-952 Gda{\'n}sk, Poland}
\affiliation{National Quantum Information Centre of Gda\'{n}sk, 81-824 Sopot, Poland}

\author{Waldemar K\l{}obus}
\affiliation{Faculty of Physics, Adam Mickiewicz University, 61-614 Pozna\'{n}, Poland}

\author{Justyna \L{}odyga}
\affiliation{Faculty of Physics, Adam Mickiewicz University, 61-614 Pozna\'{n}, Poland}

\author{Varun Narasimhachar}
\affiliation{Department of Mathematics and Statistics and Institute for Quantum Science and Technology, University of Calgary, 2500 University Drive NW, Calgary, Alberta, Canada T2N 1N4}
\affiliation{School of Physical and Mathematical Sciences and Complexity Institute, Nanyang Technological University, 50 Nanyang Ave, Singapore 639798}

\date{$\pi$ Day, 2018}

\begin{abstract}
We develop a general operational framework that formalizes the concept of conditional uncertainty in a measure-independent fashion. Our formalism is built upon a mathematical relation which we call conditional majorization. We define conditional majorization and, for the case of classical memory, we provide its thorough characterization in terms of monotones, i.e., functions that preserve the partial order under conditional majorization. We demonstrate the application of this framework by deriving two types of memory-assisted uncertainty relations: (1) a monotone-based conditional uncertainty relation, (2) a universal measure-independent conditional uncertainty relation, both of which set a lower bound on the minimal uncertainty that Bob has about Alice's pair of incompatible measurements, conditioned on arbitrary measurement that Bob makes on his own system. We next compare the obtained relations with their existing entropic counterparts and find that they are at least independent.
\end{abstract}

\maketitle 

The discovery of quantum mechanics in the early twentieth century brought about a profound change in the way we view the physical world. Among its most counter-intuitive features is the uncertainty principle, which states that there is a minimum uncertainty inherent in certain pairs of measurements \cite{Heisenberg}. The original uncertainty relations (URs) that quantify this phenomenon are widely known in the form of Robertson's formula which sets a state-dependent lower bound on the standard deviations of two operators as a measure of uncertainty \cite{Robertson1929}. It was soon realized that the content of such URs can be stated using other quantifiers of the indeterminacy of measurements. Beginning with Hirschman's formulation in terms of entropies \cite{Hirschman1957}, many different entropic URs have since been proposed in which the authors utilized the Shannon entropy and other R\'enyi entropies as measures of uncertainty \cite{Beckner1975,BiaynickiBirula1975,Deutsch1983,Wehner2010,Coles2017}. Among them, persuing a conjecture of Kraus \cite{Kraus1987}, one of the most well-known entropic uncertainty relations was derived by Maassen and Uffink \cite{Maassen1988}. It states that
\be
\label{MaassenUffink}
H(A_1)+H(A_2)	\ge\log_2\frac1c, 
\ee
where $H(\cdot)$ is the Shannon entropy of a probability distribution for an observable $A_1$ and $A_2$, and $c=\max_{a_1,a_2} |\<a_1|a_2 \>|$, where $|a_1\rangle$ and $|a_2\rangle$ are the eigenvectors of $A_1$ and $A_2$, respectively.
Nevertheless, the entropy is merely one of many ways to characterize uncertainty, and does not capture its complete operational meaning. 

Quite recently \cite{Partovi2011,Friedland2013,Puchaa2013,DSum,DSC,NPG15}, a measure-independent approach was applied to uncertainty without a memory (``non-conditional uncertainty'') in terms of a mathematical relation of \emph{majorization}. Namely, given two probability distributions $\p$ and $\q$ of some classical variable, the latter is said to be more uncertain than the former only if $\q$ is majorized by $\p$, denoted $\q \prec\p$ \cite{footnote}. Applying this operational framework on pairs (or larger sets) of outcome distributions of quantum measurements, allowed for construction of a majorization UR in the form \cite{Friedland2013,Puchaa2013}:
\be\label{UUR}
\q^{A_1} \ot \q^{A_2} \prec \bs\omega,
\ee
where $\q^{A_1}$ ($\q^{A_2}$) is the probability distribution over the outcomes of measurement of an observable $A_1$ ($A_2$) applied on an arbitrary quantum state $\rho$, and $\bs\omega\neq \left(1,0,...,0\right)$ is some \emph{constant} vector, independent of $\rho$. The left-hand side of \eqref{UUR} is more uncertain than a nontrivial fixed vector $\bs\omega$, which carries some uncertainty. The universality of such a relation lies in the fact that for any measure of uncertainty $\cMm_{\cU}$, e.g., Shannon entropy, one can get a lower bound on the joint uncertainty $\cMm_{\cU}(\q^{A_1} \ot \q^{A_2})$ by evaluating $\cMm_{\cU}$ on the fixed argument $\bs\omega$.

The study of URs is significant due to their applications in cryptographic tasks \cite{Damgaard2005,Koashi2005,DiVincenzo2004,TLGR12}, quantum correlations and nonlocality problems \cite{Oppenheim2010,Guhne2004}, and continuous-variable quantum information  processing \cite{Squeezed,CVQIP1,CVQIP2}. In some of these, it is important to consider the effect of a (sometimes adversarially controlled) \emph{memory} possibly correlated with the system. The straightforward entropic conditional uncertainty relation (CUR) in a presence of classical memory \cite{Hall1995,RenesBoileau09}:
\be
\label{Renes_clas}
H(A_1|\cR)+H(A_2|\cR)	\ge\log_2\frac1c
\ee 
can be derived from Maassen-Uffink UR \eqref{MaassenUffink}. The Shannon entropy is replaced here by the conditional Shannon entropy $H(\cdot|\cdot)$, which is calculated for distributions obtained from bipartite state $\rho^{AB}$ by measuring one of two observables $A_1,A_2$ on part $A$, and applying arbitrary measurement on part $B$, whose outcomes give rise to classical register $\cR$. The extension of inequality \eqref{Renes_clas} to quantum memory case
was conjectured in \cite{RenesBoileau09}, and proved in \cite{Berta2010}. Compared to the smooth-entropy approach of \cite{Berta2010,Tomamichel2011}, a more information-theoretic treatment of CURs \cite{coles11,Coles2012,coles2014,bcw14,KTW14,Coles2017} was developed  leading to conceptual and quantitative improvements.
However, the recent study was still restricted to entropic CURs. The possible generalization of these
(in the spirit of recently discovered majorization-based approach to non-conditional uncertainty) proved to be difficult due to lack of the notion of conditional uncertainty in abstract terms.

In this work, we resolve this problem and show how to formulate conditional uncertainty in a fully operational way, i.e., without referring to particular quantifiers of uncertainty. To this end, we define  a mathematical  relation, which we call \emph{conditional majorization}, that acts as the basis for comparing the conditional uncertainties of variables, analogously to how ordinary majorization for vectors enables the comparison of non-conditional uncertainties \cite{Friedland2013,Puchaa2013,NPG15}. It should be noted that conditional majorization is a new type of majorization relation for matrices that, to our knowledge, has not been considered so far in literature. Still, in a special case of a 1-column matrix it simplifies to standard vector majorization. Next, we apply our characterization of conditional uncertainty to derive two kinds of conditional URs with a memory in a bipartite setting: a so called \emph{monotone-based CUR} and a \emph{measure-independent CUR} {\textemdash} a counterpart to \eqref{Renes_clas}, but not involving any specific measure of conditional uncertainty.

\textit{The notion of (conditional) uncertainty.}--- Let us first outline a general approach to uncertainty relations without a memory (URs) with emphasis on the fundamental questions and challenges that inherently arise from their structure.  For a pair of measurements of observables $A_1$ and $A_2$, the UR states that the uncertainty of measurements outcomes, distributed due to probability vectors $\q^{A_1}$ and $\q^{A_2}$, respectively, is lower bounded by some constant $\omega(A_1,A_2)$. In order to construct a given UR, we then necessarily need to know first, how the uncertainty itself is defined and what class of functions can be used to measure it. Note that given two probability distributions $\p$ and $\q$, any legitimate measure of uncertainty $\cMm_{\cU}$ must naturally obey the following inequality: $\cMm_{\cU}(\q) \geq\cMm_{\cU}(\p)$, whenever $\q$ is more uncertain than $\p$. What then remains is to answer a fundamental question: when one can \emph{unambiguously state} that $\q$ is more uncertain than $\p$? In \citep{Friedland2013}, authors tackled this problem by introducing a partial order  $\p \mapsto \q$ as a tool for comparing the uncertainties of two probability distributions. Namely, with a use of just two plain axioms: i) uncertainty cannot change under permutation, ii) uncertainty cannot decrease under convex mixture (forgetting information), the uncertainty was characterized through a class of ``certainty-nonincreasing" operations, called \textit{random relabelings}, given by a set of doubly stochastic matrices $D$ such that: $\p \mapsto \q= D\p$. This relation can be equivalently expressed in terms of vector majorization denoted as $\q \prec \p$ (read: $\q$ is majorized by $\p$), if $\q$ is more uncertain than $\p$. Accordingly, any legitimate measure of uncertainty $\cMm_{\cU}$, which we call \emph{monotone}, necessarily must be monotonic under a partial order of majorization: $\q \prec \p  \Leftrightarrow \cMm_{\cU}(\q) \geq \cMm_{\cU}(\p)$. One can easily see \cite{footnote} that there is a finite set of monotones $\{\cMm_{\cU}^k \}$, given in terms of consecutive partials sums of vector components, that is sufficient as a condition for vector majorization.

Similarly, when approaching the conditional uncertainty relations in a presence of memory (CURs), we primary need to recognize the notion of conditional uncertainty and determine what functions can be used to quantify it. To this end, we shall introduce a new partial order, which qualitatively and quantitatively reflects the conditional uncertainty. In particular, we need to identify the class of operations that cannot decrease conditional uncertainty.

\textit{Classical conditional majorization.}--- We start with a case of classical memory. Alice holds a system of interest, $X$, while Bob holds another system, $\cR'$, which is the classical memory. Later on, when we apply our formalism to construct various uncertainty relations, Alice's system will result from measuring the incompatible observables on part $A$ of the initially shared quantum system  $\rho^{AB}$, whereas Bob's classical register will arise due to arbitrary chosen measurement performed by him on the quantum subsystem $B$. The classical conditional uncertainty is really a property of emerging post-measurement classical-classical (CC) state, and hence such states are the objects in our formalism.

Let us then analyze a CC state of $X\cR'$, i.e., just joint probability distributions $\p^{X\cR'}$ composing an $n \times \ell$ matrix $P=[p_{xr'}]$. In the presence of memory, we need to take into account that Bob can use his register to infer about the state of Alice's system. To include this case, we introduce the operations that generalize random relabelings. Namely, Bob is allowed to perform on his classical register $\cR'$ the arbitrary trace-nonincreasing operations (sub-stochastic maps $R^{(j)}$) that sum up to row-stochastic map, and Alice's allowed action (random relabeling given by $n\times n$ doubly stochastic matrix $D^{(j)}$) on $X$ can depend on Bob's output ($j$) -- it can actually be only a mental action of Bob, but it is convenient to assign it to Alice. This results in the following general form of a ``conditional certainty-nonincreasing" operations, which we call \textit{classical-conditioned random relabelings} (cf. \appendix~\ref{subsec1A}-\ref{subsec1D}):
\be\label{qform}
P\mapsto Q=\sum_{j} D^{(j)}P R^{(j)},
\ee
where $j$ can run over an arbitrary number of values. Since the effective transformation acting on $X$ is always doubly stochastic, the number of rows of $P$ stays intact, however, the arbitrary classical channels allowed on $\cR'$ can change the number of columns in $P$, possibly transforming $\cR'$ to a different classical system $\cR$, leading to an $n \times m$ probability matrix  
$Q=[q_{xr}]$.

Now, whenever $Q$ can be obtained from $P$ as in \eqref{qform}, we say that the probability matrix $Q$ \textit{is conditionally majorized} by $P$ denoted
\be\label{condmaj}
Q\prec_c P,
\ee
which means that the variable $X$ is more uncertain in the CC state $X\cR$ than in the state $X\cR'$, conditioned upon access to side information. Note that the relation of conditional majorization does not constitute a total order (similarly as ordinary majorization), because given two matrices, not always the mapping of the form \eqref{qform} exists. 
Conveniently, the relation \eqref{condmaj} is easily checkable by linear programming methods as shown in \appendix \ref{subsec1E}. 

In this way, we have arrived at the definition of being  more conditionally uncertain: $Q$ is more conditionally uncertain than $P$ when $Q\prec_c P$. In particular, when $Q$ and $P$ are both product, then this reduces to ordinary majorization of their $X$ marginals. On the other hand, when $P$ is maximally correlated, then it does not exhibit conditional uncertainty at all, and all $Q$'s are conditionally majorized by such $P$.

\textit{Quantum conditional majorization.}---  We can analogously define partial order in the case of quantum memory. As before, Alice holds the classical register $X$, whereas Bob now owns a quantum memory $B'$. Together they share a classical-quantum (CQ) state in the form: $\sigma\!=\!\sum_{x=1}^{n}p_x\proj x^X\!\otimes\sigma_{x}^{B'}$, where $p_x$ is the marginal probability of Alice's outputs $x$, $\{\ket x\}_x$ denotes an orthonormal basis consisting of ``classical outcome flags'' on the classical register $X$ and $\sigma_{x}^{B'}$ describes the state of quantum memory system. The conditional majorization is then defined by a class of operations, which we call \textit{quantum-conditioned random relabelings}, that relates two CQ states rather than two CC states (as stated for classically correlated memory): $
\sigma\mapsto\tau=\sum_{x,y=1}^{n}\sum_j p_xD^{(j)}_{yx}|y\rangle\langle y|^X\otimes[\cM^{(j)}(\sigma_{x}^{B'})]^B$,
where $\cM^{(j)}$ represents an arbitrary quantum operations and is a trace-nonincreasing completely positive map, such that the map $\cM(\cdot)\equiv\sum_j\cM^{(j)}(\cdot)$ is trace-preserving. We then say that a CQ state $\tau$ is more conditionally uncertain than $\sigma$, and denote it as $\tau\prec_c \sigma$.

\textit{Characterization of conditional majorization.}--- 
Now, we will present necessary and sufficient conditions for classical-conditional majorization in the form of \emph{monotones}, i.e., real-valued functions that preserve the partial order under conditional majorization (see \appendix \ref{subsec1F} for more details):
\begin{theorem}
For joint probability distributions $Q^{X\cR}$ and $P^{X\cR'}$, we have that $Q^{X\cR} \prec_c P^{X\cR'}$ iff for all convex symmetric functions $\Phi$ the following inequality holds:
\be\label{theorem1}
\sum_{r}q_{r}\Phi(\q^{X|r}) \leq \! \sum_{r'}p_{r'}\Phi(\p^{X|r'}),
\ee
where $q_{r}$ is the marginal probability distribution of $\cR=r$, and $\q^{X|r}$ is the conditional probability distribution of $X$ given $\cR=r$.
\end{theorem}

The theorem implies that the proper monotones $\cMm_{\cC\cU}$ that quantify conditional uncertainty are given by the averages of arbitrary convex functions $\Phi$, as seen in Eq. \eqref{theorem1}. The minimal set of monotones that determines the partial order under conditional majorization is not identified, nevertheless there exists a set of monotones that is sufficient as a condition for conditional majorization. It can be achieved by restricting the functions $\Phi$  in Eq. \eqref{theorem1} to the subset of functions $\Phi_A$ in the form (\appendix \ref{subsec1F}):
\be\label{monotoneA}
\Phi_A(\q^{X|r})\equiv \max_{k} \; (\a_{k})^{\downarrow} \bullet (\q^{X|r})^{\downarrow},
\ee
where $\a_{k}$ is the $k^{\text{th}}$ column of a row-stochastic matrix $A=[a_{xk}]$, and down arrow denotes the alignment of vector components in nonincreasing order. Interestingly, there exists a finite number of conditions for conditional majorization, provided by the standard linear programming techniques, which are however not given in terms of monotones and are rather nonintuitive.

The developed concept of conditional uncertainty \textit{via} conditional majorization allows now for the proper construction of CURs. In particular, two different approaches are possible. On one hand, one can derive CURs based on some monotone $\cMm_{\cC\cU}$, given in the generic form: $\cMm_{\cC\cU}(Q^{X\cR})\leq\gamma$, where $\cMm_{\cC\cU}(Q^{X\cR})=\sum_{r}q_{r}\Phi(\q^{X|r})$  is obtained by choosing a fixed function $\Phi$ (or $\Phi_A$ given in Eq. \eqref{monotoneA}), $Q^{X\cR}$ emerges from quantum mechanical measurements and a bound $\gamma$ depends only on the observables. In the second approach, one can use directly the very definition of conditional uncertainty given by conditional majorization and derive the universal (i.e., measure-independent) CUR in the form $Q^{X\cR} \prec_c \Omega$, a memory-assisted generalization of \eqref{UUR}. In the following, we will start by choosing a single, simple monotone and derive a CUR based on it. Subsequently, we will use the obtained relation to provide a whole family of universal CURs.

In concrete applications, $X$ will be register built of  outcomes of observables 
$A_1$ and $A_2$ whose uncertainty we are going to describe. 
It can be given by $x=(a_1, a_2)$ (giving rise to CUR of ``tensor product" type)
but also by  $x=a$ where $a$ runs through union of the observables outputs ("direct sum" type). Our examples below will be of the first kind.

\textit{Application 1: CUR based on a monotone.}--- Let us first note that for a classical memory, a given CUR can be derived from a corresponding
UR, in a similar way as the CUR \eqref{Renes_clas} was derived from Maassen-Uffink UR \eqref{MaassenUffink}, provided the UR is given in terms of a convex function. Consider the UR in the form
\be
\label{URgeneral}
\cMm_{\cU}(\q)\equiv \Phi(\q^{A_1}\ot \q^{A_2}) \leq \gamma,
\ee 
where $q^{A_1} \left( q^{A_2}\right) $ results from measuring $A_1$ ($A_2$) on a quantum system and $\gamma$ is a constant dependent only on the observables $A_1,A_2$. Then, the corresponding CUR appears as
\be
\label{CURgeneral}
\cMm_{\cC\cU}(Q^{X\cR}) \equiv \sum_r q_r \Phi(\q^{A_1|r}\ot \q^{A_2|r}) \leq \gamma
\ee
with $q_{r}$ denoting the marginal probability distribution of $\cR=r$, and $\q^{A_1|r}$ ($\q^{A_2|r}$) describing the probability distribution of an observable $A_1$ ($A_2$) conditioned on $\cR=r$.

Here,  the measurement scenario is established in a bipartite setting, where two parties, Alice and Bob, initially share two copies of a pure state: $\Psi^{A_1B_1}\otimes\Psi^{A_2B_2}$. On the first copy, Alice and Bob perform POVMs $\left\{M^{A_1}_{a_1}\right\}$ and $\left\{M^{R}_{r}\right\}$, respectively, whereas on the second copy Alice measures $\left\{M^{A_2}_{a_2}\right\}$ and Bob repeats the measurement performed on the first copy. The
resulting measurement statistics are given by
\ben
\label{elements1} 
q_{a_1 r}&=&\tr\left[ \left( M^{A_1}_{a_1}\otimes M^{R}_{r} \right)\Psi^{A_1B_1} \right],\\
\label{elements2} 
q_{a_2 r}&=&\tr\left[ \left( M^{A_2}_{a_2}\otimes M^{R}_{r} \right)\Psi^{A_2B_2} \right].
\een

By combining the Alice's measurement outcomes into a single variable $x\equiv(a_1,a_2)$, one obtains the matrix $Q^{X\cR}=[q_{xr}]$ with 
\be
\label{q_elements}
q_{xr}=q_r q_{a_1|r} q_{a_2|r}.
\ee
Then, using the result \eqref{UUR} and \eqref{CURgeneral}, we derive (see \appendix \ref{subsec2A}) a monotone-based CUR in the form:
\be
\label{eq:CUR_major}
\sum_r q_r \Phi_l(\q^{X|r})\leq \eta_l,
\ee
where $\Phi_l(\q^{X|r})$ is the sum of $l$ largest probabilities of the distribution $\q^{X|r}$, and the functions $\Phi_l$ are obtained from  Eq. \eqref{monotoneA} by choosing  $A= (\underbrace{1,\ldots, 1}_{l},0,\ldots 0)^T$.
For $l=1$, an explicit form of $\eta_l$ in terms of observables was provided in \cite{Puchaa2013,Friedland2013}, and is given by $\eta_1=  \tfrac14 (1+c)^2$ 
with $c$ as in Eq. \eqref{MaassenUffink}.

{\it Remark.} Note that in Eq. \eqref{eq:CUR_major} one cannot take just a function that quantifies non-conditional uncertainty, and condition on it.
For example R\'enyi entropy for $\alpha >1$ is a valid quantifier of non-conditional uncertainty because it is monotonic under majorization, but it violates conditional majorization order because it is not convex.

\begin{figure}[t!]
  \centering
  \includegraphics[width=0.25\textwidth]{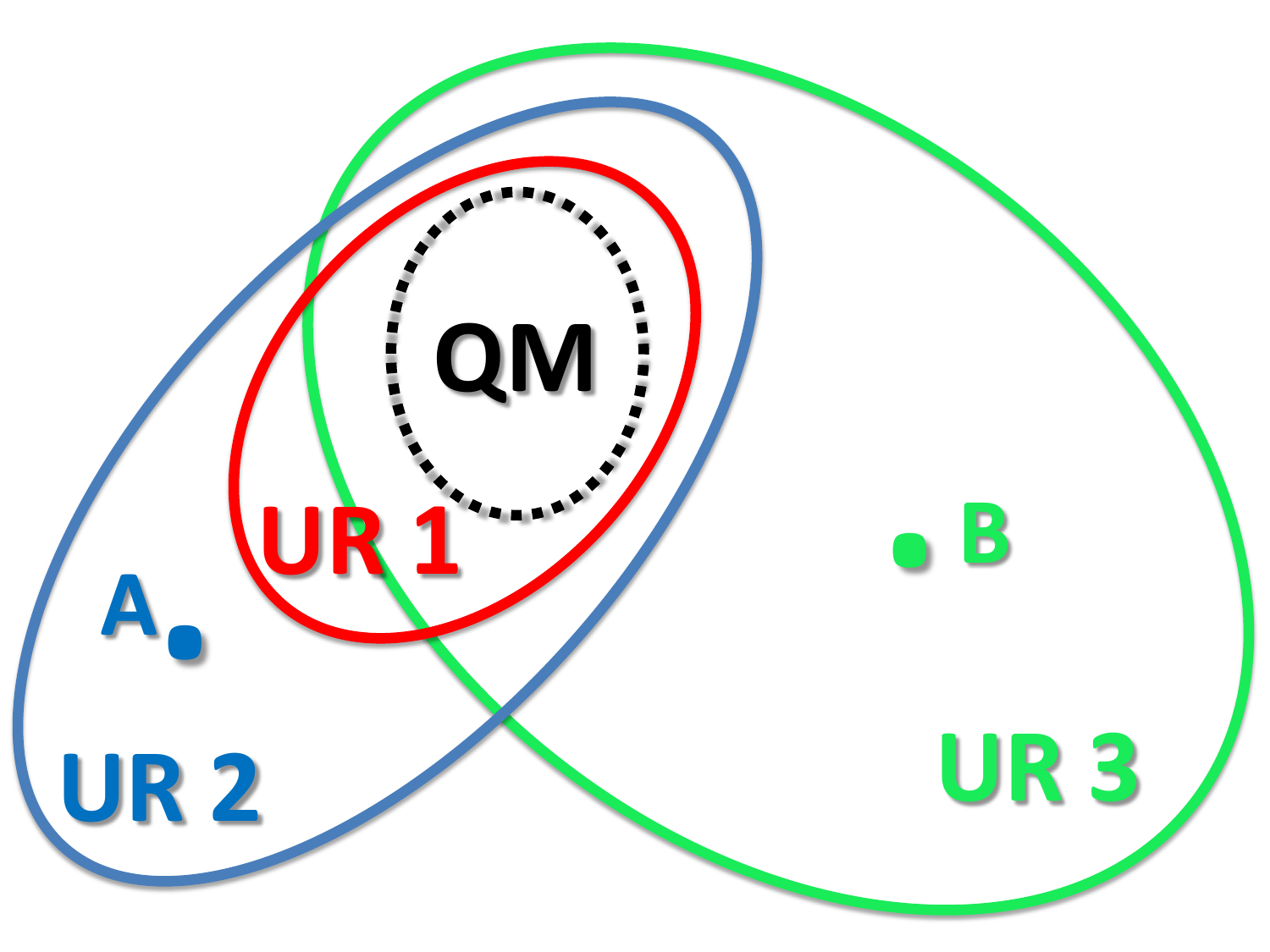}\\
  \caption{How to compare different uncertainty relations. The solid line sets represent three distinct URs: UR1, UR2 and UR3. The statistics allowed by a given UR are represented by points contained in the corresponding set, whereas the excluded statistics do not belong to the set. The statistics originating from quantum-mechanical measurements obey all the URs and form a dashed line subset QM of all solid line sets. We say that UR1 is stronger than UR2 since UR1$\subset$ UR2, i.e., all statistics allowed by UR2 are also allowed by UR1, but UR1 excludes some statistics allowed by UR2 (e.g., point A). On the other hand, UR2 is independent of UR3 if UR2 $\cap$ UR3 $\neq$ UR2 $\neq$ UR3, i.e., it excludes some statistics allowed by UR3 (point B), but also allows some statistics excluded by UR3 (point A). Naturally, the same reasoning applies to comparing different CURs.}
  \label{scheme}
\end{figure}

\textit{Application 2: Measure-independent CUR based on conditional majorization.}--- 
The CUR obtained from monotone $\Phi_1$ will now serve  
as a tool to derive a family of CURs given by a conditional majorization relation, hence independent of any particular measure. For a previously examined measurement scenario, we lower-bound the conditional uncertainty of $X|\cR$ in terms of conditional majorization. Namely, we construct a ``matrix-valued'' upper bound on $Q^{X\cR}=[q_{xr}]$ of Eq. \eqref{q_elements} under the ``$\prec_c$'' partial order by finding some $\Omega\in\cC\cC^n$, such that $Q^{X\cR} \prec_c \Omega$ (cf. \appendix \ref{subsec2B}):  

\begin{theorem}
Define $c\equiv\max_{a_1,a_2} |\<a_1|a_2 \>|$, where $a_1$ and $a_2$ are eigenvectors of Alice's $A_1$ and $A_2$ measurements, respectively.
Then for arbitrary  state $\rho^{AB}$ and an arbitrary chosen measurement on system $B$ we have
\ben\label{eq:Omega}
Q^{X\cR} \prec_c \Omega,
\een
with $Q^{X\cR}=[q_{xr}]$ defined in Eq. \eqref{q_elements}, and
\be\label{coromega}
\Omega =
\left[\begin{array}{ccccc}
\alpha              & 0                    &  \ldots& 0\\
(1-\alpha) \omega_1 & (1-\alpha) \omega_2  & \ldots & (1-\alpha) \omega_n
\end{array}
\right]^T,
\ee
where
$\bs\omega=(\underbrace{\beta, \ldots, \beta}_{l}, 1 - l \beta, 0\ldots 0)$ with $l$ being the largest integer such that $\beta l \leq 1$ and $\alpha,\beta$ satisfy $\alpha \beta \leq \frac14 (1 +c)^2$.
\end{theorem}

The relation \eqref{eq:Omega} constitutes a nontrivial CUR when $\alpha$ and $\beta$ are strictly less than $1$.
This happens for ${\frac14(1 +c)^2<\alpha<1}$. If we set $\alpha=1$, and take the systems $B$ to be trivial, the CUR reduces to majorization UR of \eqref{UUR} with $\bs\omega=(\omega_1,\omega_2,...)$, where $\omega_1=\frac14(1+c)^2$, $\omega_2= 1- \omega_1$, and $\omega_i=0$ for $i>2$.  

\begin{figure}[t!]
  \centering
  \includegraphics[width=0.35\textwidth]{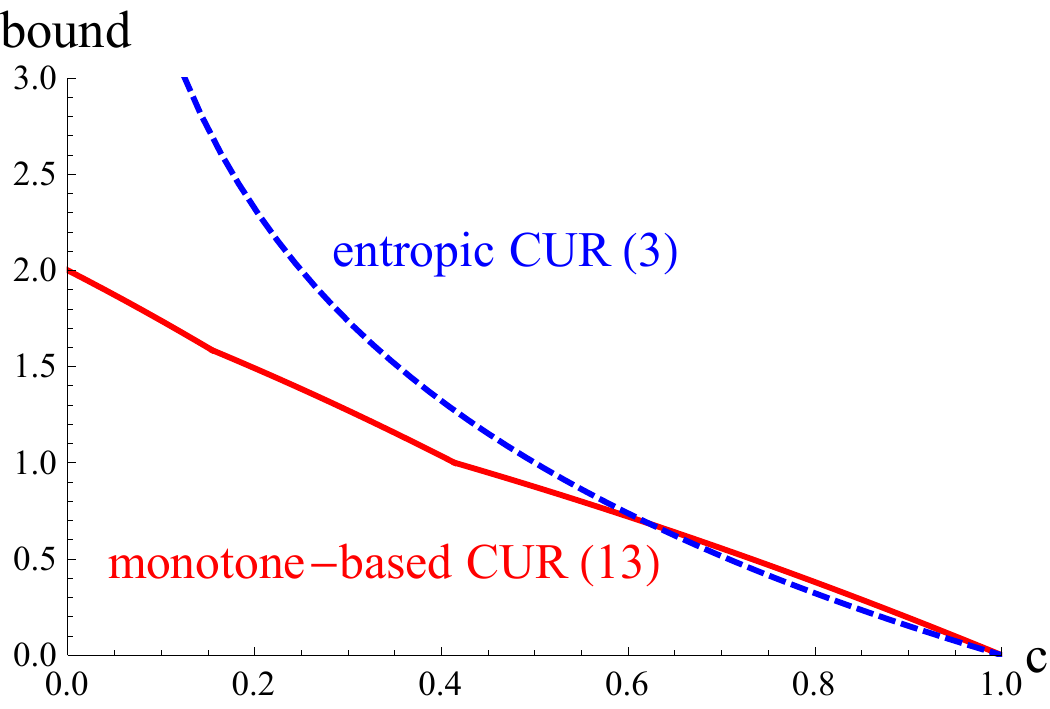}\\
  \caption{Comparison of entropic CUR \eqref{Renes_clas} with its counterpart implied by monotone-based CUR \eqref{eq:CUR_major}. The bound for entropic CUR constitutes the RHS of \eqref{Renes_clas}, whereas the bound for monotone-based CUR is given by minimal conditional entropy for probability distributions which satisfy \eqref{eq:CUR_major},  see \appendix \ref{subsec3B} for details.}
  \label{plot}
\end{figure}

\textit{Comparison with existing CURs.}--- One might wonder what motivation stands behind the studying of universal uncertainty relations in the first place. Let us then delve into the fundamental and, by all means, interesting insights of the obtained universal CUR in Theorem 2. First of all, the universal CUR is a versatile tool with which one can derive a CUR for any arbitrary measure of conditional uncertainty, even for the functions not yet considered in this context in literature. Secondly, it does not merely say that some measure of uncertainty is non-trivially limited, but it gives insight into the \emph{structure} of quantum-mechanical distributions.

It might seem that the price for the universality of the measure-independence is that the resulting CURs would come out weaker comparing to existing CURs based on particular measures (like entropy). However, remarkably, it is not the case. To see it, we first explain in Fig. \ref{scheme} how to compare various (C)UR's. It turns out that our universal CUR \eqref{eq:Omega} is \emph{independent} of entropic CUR 
\eqref{Renes_clas}, see \appendix \ref{subsec3A}. In particular, it can be shown that \eqref{eq:Omega} excludes statistics with too large conditional  probabilities of certain events (one for observable $A_1$, and the other for observable $A_2$), while such statistics are allowed by entropic CUR  \eqref{Renes_clas}. Moreover, for states such that the memory is classical, the monotone-based CUR \eqref{eq:CUR_major}  with $l=1$ appears to be stronger than \eqref{Renes_clas} for some range of parameter $c$ (Fig. \ref{plot}).

\textit{Concluding remarks.} 
In this work, we have introduced an operational framework that through a mathematical relation ``$\prec_c$'' called conditional majorization defines conditional uncertainty{\textemdash}that is, uncertainty of a classical variable (e.g., a measurement outcome) conditioned on access to a correlated memory (either classical or quantum). In applications, we focused on the case of classical memory, and constructed two conditional uncertainty relations in a bipartite setting: a monotone-based CUR and a universal measure-independent CUR. Both of them seem to beat exisiting entropic CUR to some extent.

The framework of conditional uncertainty is however more general and versatile, and can be used in many other scenarios. Among others, our approach applied here to the ``tensor product'' type majorization \eqref{URgeneral}-\eqref{CURgeneral}, can be easily applicable to the direct sum majorization \cite{DSum,DSC,NPG15} as well. Also, our framework opens a new path of research, with many challenges, to mention a possible generalization to cover the case of explicit dependence on entanglement as in \cite{Berta2010}. Another important project to be undertaken is the application of our methods to concrete cryptographic tasks, such as key distribution and coin-tossing. Replacing the existing, entropy-based methods with majorization-based methods is likely to improve the analysis of the single-shot or finite-size case of such tasks, possibly beyond the improvement already afforded by the smooth R\'enyi entropy calculus.

\begin{acknowledgments}
This work is supported by ERC Advanced Grant QOLAPS and National Science Centre grants: Maestro DEC-2011/02/A/ST2/00305 and OPUS 9. 2015/17/B/ST2/01945. VN acknowledges financial support from the Ministry of Education of Singapore, the National Research Foundation (NRF), NRF-Fellowship (Reference No: NRF-NRFF2016-02) and the John Templeton Foundation (Grant No 54914).
\end{acknowledgments}

\setcounter{equation}{0}
\renewcommand\theequation{A\arabic{equation}}
\section*{Appendix}
\section{Classical-conditional majorization}
\label{sec1}
Consider the following scheme: Alice holds a system of interest, $X$; Bob holds another system, $\cR'$, which is the classical memory. The state they share is a basic object of our formalism and is given by a classical-classical state $P^{X\cR'}\in\cC\cC^n$, which we will simply denote as $P$. Let then $P=[p_{xr'}]$ be an $n \times \ell$ matrix whose components are the joint probabilities $p_{xr'}$. Hereafter, we use the notation $p_{r'}\equiv\sum_xp_{xr'}$ for the marginal probability of $\cR'=r'$, and $\p^{X|r'} \equiv \p^{|r'}\equiv\left(p_{x|r'}\right)_x=p_{r'}^{-1}\left(p_{xr'}\right)_x$ for the conditional distribution of $X$ given $\cR'=r'$. The same notation applies to a $n \times m$ probability matrix $Q^{X\cR} \equiv Q=[q_{xr}]$. Denote by $\R^{n\times\ell}_{+,1}$ the set of all $n\times\ell$ row-stochastic matrices, and by $\R^{n\times\ell}_{+}$ the set of all $n\times\ell$ matrices with non-negative entries. Then, $P\in\R^{n\times\ell}_{+}$ and $Q\in\R^{n\times m}_{+}$.
\subsection{Condition for the conditional majorization}
\label{subsec1A}
In the main text, in Eq. (4) we defined the class of operations, called \textit{classical-conditioned random relabelings} (CCR), that do not decrease classical conditional uncertainty. This led to the following condition for the classical-conditional majorization relation $Q\prec_cP$:
\be\label{qforma}
Q=\sum_{j} D^{(j)}P R^{(j)},
\ee
where each $D^{(j)}$ is an $n\times n$ doubly stochastic matrix, and each $R^{(j)}$ is an $\ell\times m$ matrix of non-negative entries, with $\sum_jR^{(j)}$ row-stochastic.

\begin{remark}
Arbitrarily reordering the elements of each column of $P$ is a CCR: Just choose $D^{(j)}=\Pi_j$ (permutations), and $R^{(j)}=e_{jj}$, i.e., a matrix with zeros everywhere except the $(j,j)$ element, which is chosen to be 1.
\end{remark}
In component-wise form it reads as:
\be\label{components}
q_{xr}=\sum_{j}\sum_{r'=1}^\ell \left(D^{(j)}\p_{r'}\right)_{x}R^{(j)}_{r'r}.
\ee
Here $(D^{(j)}\p_{r'})_x$ is the $x^\mathrm{th}$ component of the vector $D^{(j)}\p_{r'}$, where $\p_{r'}$ is the $(r')^\mathrm{th}$ column of $P$ (we will similarly use $\q_r$ for the columns of $Q$). To simplify this relation, we denote
\be
t_{r'r}\equiv \sum_{j}R^{(j)}_{r'r}\;\;\text{ and }\;\;D^{(r',r)}\equiv\sum_{j}\frac{R^{(j)}_{r'r}}{t_{r'r}}D^{(j)}.
\ee
Note that $D^{(r',r)}$ is an $n\times n$ doubly stochastic matrix and $T=[t_{r'r}]$ is an $\ell\times m$ row-stochastic matrix. With these notations the relation in~\eqref{components} becomes
\be\label{qpa}
\q_{r}=\sum_{r'=1}^\ell t_{r'r}D^{(r',r)}\p_{r'}.
\ee

\begin{lemma}\label{lem}
$Q\prec_{c}P$ if and only if there exists a set of $n\times n$ doubly stochastic matrices $D^{(r',r)}$ ($\ell m$ in number) and an $\ell\times m$ row-stochastic matrix $T\equiv[t_{r'r}]$ such that the columns of $P$ (denoted $\p_{r'}$) are related to those of $Q$ ($\q_r$) through \eqref{qpa}.
\end{lemma}

\begin{proof}
We already proved that $Q\prec_{c}P$ implies the existence of $D^{(r',r)}$ and $T$ with the desired properties, satisfying \eqref{qpa}. It remains to show that if such matrices exist, then there exist $D^{(j)}$ and $R^{(j)}$ that satisfy~\eqref{qforma}. Indeed, since $D^{(r',r)}$ are doubly stochastic, we can write them as $D^{(r',r)}=\sum_{j}s^{(j)}_{r'r}\Pi^{(j)}$, where $\Pi^{(j)}$ are permutation matrices and $\sum_js^{(j)}_{r'r}=1$. Hence, taking $R^{(j)}$ to be the matrix whose components are $t_{r'r}s^{(j)}_{r'r}$, and $D^{(j)}$ the permutation matrices $\Pi^{(j)}$, completes the proof.
\end{proof}

\subsection{Standard form}\label{subsec1B}
In the following, we introduce the notion of a standard form that simplifies a characterization of classical-conditional majorization.

We say that $Q$ and $P$ are \emph{conditionally equivalent}, and write $Q\sim_c P$, if
$Q\prec_cP$ and $Q\succ _c P$. Identifying cases of such equivalence leads to further simplification, resulting in the following standard form:
\begin{definition}[Standard Form]\label{stf}
Let $Q=[q_{xr}]$ be an $n\times m$ joint distribution matrix. Its standard form $Q^{\downarrow}=[q^{\downarrow}_{xs}]$ is an $n\times\tilde{m}$ matrix, where $\tilde{m}\le m$. To obtain $Q^\downarrow$, we apply the following transformations on $Q$:
\begin{enumerate}
\item Reordering within columns: Arrange the elements of each column of $Q$ in non-increasing order. Since this is a reversible CCR, we can just assume WLOG that $Q$ has this form. That is, for all $r=1,...,m$, $q_{1r}\ge q_{2r}\dots\ge q_{nr}$.
\item Combining proportional columns: If two columns of $Q$, say those corresponding to $r$ and $\tilde{r}$, are proportional to each other (i.e., $\q^{|r}=\q^{|\tilde{r}}$), then we replace both by a single column $(q_r+q_{\tilde{r}})\q^{|r}$. We do this until no two columns are proportional. The resulting matrix $\tilde Q$ contains $\tilde{m}\le m$ columns, each of the form $\tilde q_s\tilde\q^{|s}$.
\item Reordering the columns: Reorder the $\tilde{m}$ columns of $\tilde Q$ in non-increasing order of $\tilde q_s$. If there are ties, resolve them by ranking the tied columns in non-increasing order of their first component $\left(\tilde q_{1s}\right)$. If there remain ties, we rank by non-increasing $\left(\tilde q_{1s}+\tilde q_{2s}\right)$, and so on.
\end{enumerate}
The resulting $n\times\tilde{m}$ matrix is the final standard form $Q^\downarrow$. 
\end{definition}

Below, we show that $Q^\downarrow\sim_cQ$.
\begin{lemma}\label{sim}
Consider an $n\times m$ probability matrix $Q=[\q_1,...,\q_m]$ such that
one of the columns, say $\q_1$, is a multiple of another column, say $\q_2$.
That is, there exists non-negative real number $\lambda$ such that $\q_1=\lambda\q_2$.
Then,
\be
Q\sim_c Q'\equiv [(1+\lambda)\q_2,\q_3,...,\q_m],
\ee
where $Q'$ is $n\times (m-1)$ probability matrix.
\end{lemma}
\begin{proof}
Let $T=[\bt_1,...,\bt_{m-1}]$ be the following $m\times (m-1)$ row stochastic matrix
\be
T=\begin{pmatrix}
1 & 0 & \cdots & 0\\
\; & & \; I_{m-1} & \;
\end{pmatrix},
\ee
where $I_{m-1}$ is the $(m-1)\times(m-1)$ identity matrix. Note that $QT=Q'$ and therefore $Q'\prec_c Q$.
Let $U=[\u_1,...,\u_{m-1}]$ be the following $(m-1)\times m$ row stochastic matrix
\be
U=\begin{pmatrix}
\frac{\lambda}{1+\lambda} & \frac{1}{1+\lambda} & 0 & \cdots & 0\\
0 & 0 & 1 & \cdots & 0\\
\vdots & \vdots & \; & \ddots  & \vdots\\
0 & 0 & \; & \; & 1
\end{pmatrix}.
\ee
Note that $Q'U=Q$ and therefore $Q\prec_c Q'$. Hence, $Q\sim_c Q'$.
\end{proof}

\subsection{Characterization of conditional majorization through the lower-triangular matrix L}\label{subsec1C}
Hereafter, we will (often implicitly) assume all states to be in their standard form, without loss of generality. This enables an elegant characterization of conditional majorization through the $n\times n$ lower-triangular matrix
\be
\label{Lmatrix}
L=
\begin{pmatrix}
1 & 0 & 0 & \cdots & 0\\
1 & 1 & 0 & \cdots & 0\\
1 & 1 & 1 & \cdots &0\\
\vdots & \vdots & \; & \ddots & \vdots\\
1 & 1 & 1 & \cdots & 1
\end{pmatrix}.
\ee
\begin{lemma}\label{lem3}
For $Q$ and $P$ in the standard form, $Q\prec_{c} P$ if and only if there exists a row-stochastic matrix $T$ such that
\be\label{mqpt}
LQ \le LPT,
\ee
where the inequality is entry-wise.
\end{lemma}
\begin{proof}
Since all the $D^{(r',r)}$ in~\eqref{qpa} are doubly stochastic, we get that
\be\label{par}
\sum_{x=1}^{k}q_{xr}\leq \sum_{r'=1}^\ell t_{r'r}\sum_{x=1}^{k}p_{xr'},\;\;\;\forall\;k=1,...,n.
\ee
Note that by taking the sum over $r$ on both sides, we get that the marginal distributions $\q=(\sum_rq_{xr})_x$ and $\p=(\sum_{r'}p_{xr'})_x$ satisfy $\q\prec\p$.
Denote by $\tilde{Q}$ and $\tilde{P}$ the matrices whose components are
\be
\tilde{q}_{kr}=\sum_{x=1}^{k}q_{xr}\;\;\text{ and }\;\;\tilde{p}_{kr'}=\sum_{x=1}^{k}p_{xr'}.
\ee
That is,
\be
\tilde{Q}=LQ\;\;\text{ and }\;\;\tilde{P}=LP,
\ee
with $L$ defined in \eqref{Lmatrix}.
Due to Eq.~\eqref{par}, if $Q\prec_{c} P$ then there exists a row-stochastic matrix $T$ such that
$\tilde{Q}\leq\tilde{P}T$.
Conversely, suppose there exists a row stochastic matrix $T$ satisfying~\eqref{mqpt}.
Denote $A\equiv PT$ and the components of $A$ by
\be
a_{kr}=\sum_{r'=1}^{\ell}p_{kr'}t_{r'r}.
\ee
The condition $\tilde{Q}\leq\tilde{P}T$ is equivalent to $\q_{r}\prec_{\rm w}\a_{r}$, where
$\a_{r}=(a_{kr})_{k}$, and the symbol $\prec_{\rm w}$ stands for weak majorization (i.e., instead of equality, we have $\sum_{k=1}^{n}q_{kr}\leq\sum_{k=1}^{n}a_{kr}$).
It is known (see e.g., \cite{Marshall2011}) that $\q_{r}\prec_{\rm w}\a_{r}$ if and only if
there exists a non-negative (entry-wise) matrix $S^{(r)}$ and a doubly stochastic matrix $D^{(r)}$ such that
$\q_{r}=S^{(r)}\a_{r}$ and $S^{(r)}\leq D^{(r)}$ entry-wise.
Therefore, there exists doubly stochastic matrix $D^{(r)}$ such that
$\q_{r}\leq D^{(r)}\a_{r}$. In components it reads
\be
q_{xr}\leq \sum_{k=1}^{n}D^{(r)}_{xk}a_{kr}=\sum_{k=1}^{n}\sum_{r'=1}^{\ell}D^{(r)}_{xk}p_{kr'}t_{r'r}\equiv q_{xr}'.
\ee
However, since the components $q_{xr}$ and $q_{xr}'$ both sums to one, the condition $0\leq q_{xr}\leq q_{xr}'$ implies that $q_{xr}=q_{xr}'$. Hence,
this equation is equivalent to~\eqref{qpa} with $D^{(r',r)}\equiv D^{(r)}$, and from Lemma~\ref{lem} it follows that $Q\prec_{c} P$.
This completes the proof.
\end{proof}

\begin{remark}
Note that the proof of Lemma~\ref{lem3} also implies that
$Q\prec_c P$ if and only if there exist $m$ doubly stochastic matrices $D^{(r)}$, and a row stochastic matrix $T$ such that
\be\label{qp1}
\q_{r}=D^{(r)}\sum_{r'=1}^\ell t_{r'r}\p_{r'},\;\;\;\forall\;r =1,...,m.
\ee
This is a simpler version of~\eqref{qpa}. For any set of $m$ doubly stochastic matrices $\cD=\{D^{(1)},...,D^{(m)}\}$, and an $n\times m$ matrix $A$, whose columns are $\a_z$, we define
\be
\label{sqr_brackets1}
\cD[A]:=[D^{(1)}\a_1,...,D^{(m)}\a_m].
\ee
With these notations
\be
\label{sqr_brackets2}
Q\prec_c P\;\;\iff\;\;Q=\cD[PT]
\ee
for some set of $m$ doubly stochastic matrices $\cD$ and a row stochastic matrix $T$.
\end{remark}

\subsection{Properties of the conditional majorization relation}\label{subsec1D}
The standard form (see Definition \ref{stf}), together with the simplification afforded by Lemma~\ref{lem3}, leads to some nice properties of the conditional majorization relation:
\begin{theorem}\label{thm1}
Let $P$, $Q$, $R$ be three probability matrices. Then,
\begin{align}
\label{reflexivity}
&\text{\rm Reflexivity:}\;\;Q\prec_{c} Q.\\
\label{transitivity}
&\text{\rm Transitivity:}\;Q\prec_{c} P \text{ and }P\prec_c R\;\Rightarrow\;
Q\prec_c R.\\
\label{antisymmetry}
&\text{\rm Antisymmetry:} \;Q\prec_{c} P\text{ and }P\prec_c Q\;\Rightarrow\;Q^{\downarrow}\!=\!P^{\downarrow}.
\end{align}
That is, $\prec_c$ is a partial order with respect to the standard form.
\end{theorem}
\begin{proof}
Reflexivity \eqref{reflexivity} and transitivity \eqref{transitivity} of $\prec_c$ follow directly from its definition in~\eqref{qforma}. To prove antisymmetry \eqref{antisymmetry}, suppose $Q\prec_c P$ and $P\prec_c Q$ (i.e., $Q\sim_c P$), and suppose $Q$ and $P$ are given in their standard form; that is, we assume $Q=Q^{\downarrow}$ and $P=P^{\downarrow}$. From Lemma~\ref{lem3} we know that there exist stochastic matrices $T$ and $R$ such that $LQ\leq LPT$ and $LP\leq LQR$. Combining these two inequalities gives
\be
LQ\leq LQRT\;\;\text{ and }\;\;LP\leq LPTR.
\ee
Since $RT$ is a row stochastic matrix, the sum of the columns of $LQ$ and $LQRT$ are the same. Therefore, we must have $LQ=LQRT$, and using similar arguments we also get $LP=LPTR$. Since $L$ is invertible, this in turn gives
\be\label{hhh}
Q=QRT\;\;\text{ and }\;\;P=PTR.
\ee
We next prove, that $RT$ and $TR$ must be the identity matrices.
\begin{lemma}
\label{lem5}
Let $A$ be an $m\times m$ row stochastic matrix, and $Q=Q^{\downarrow}$ is an $n\times m$ matrix with non-negative components. Then,
\be
Q=QA\;\;\Rightarrow\;\;A=I_m.
\ee
\end{lemma}
The proof is by induction over $m$. For $m=1$,
$Q$ is an $n\times 1$ column matrix. Then, $A$ is a $1\times 1$ row stochastic matrix,
i.e., the number 1. Therefore the Lemma \ref{lem5} holds for $m=1$. Next, suppose the lemma holds for all $n\times m$ non-negative matrices $Q^{(m)}$ (in their standard form); i.e., if $Q^{(m)}=Q^{(m)}A^{(m)}$ for some $m\times m$ row stochastic matrix
$A^{(m)}$, then $A^{(m)}=I_m$. We need to show that the same holds for all $n\times (m+1)$ non-negative matrices $Q^{(m+1)}$ (in their standard form).
Indeed, let $A^{(m+1)}$ be an $(m+1)\times (m+1)$ row stochastic matrix, and denote
\be
Q^{(m+1)}\!=\!\left(\begin{array}{c|c}
 \! \!Q^{(m)}\! & \q_{m+1}\!\!
 \end{array}\right)\;\text{and}\;
A^{(m+1)}\!=\!\left(\begin{array}{c|c}
  \!\!\!A^{(m)}\! & \!\v\!\! \\
  \hline
  \!\!\!\u^T\! & \!1\!-\!u\!\!
 \end{array}\right),
\ee
where $\u,\v\in\R^{m}_{+}$ and $u$ is the sum of the components of $\u$. Note that while $A^{(m)}$ above is non-negative it is not necessarily row stochastic. More precisely, the sum of the columns of $A^{(m)}$ is $\e-\v$, where $\e=(1,...,1)^T$. Now,
suppose that $Q^{(m+1)}=Q^{(m+1)}A^{(m+1)}$. With the notations above this is equivalent to
\begin{align}\label{mom}
& Q^{(m)}=Q^{(m)}A^{(m)}+\q_{m+1}\u^T,\nonumber\\
& \q_{m+1}=Q^{(m)}\v+(1-u)\q_{m+1}.
\end{align}
The second equation is equivalent to $u\q_{m+1}=Q^{(m)}\v$. Therefore, if $u=0$ then $\v=0$ since $Q^{(m)}$ has no zero columns. Clearly, in this case we also have $\u=0$ so that $Q^{(m)}=Q^{(m)}A^{(m)}$
and therefore $A^{(m)}=I_{m}$ from the assumption of the induction. This also gives $A^{(m+1)}=I_{m+1}$.
We therefore assume now that $u>0$ and thus $\v\neq 0$.
Substituting $\q_{m+1}=\frac{1}{u}Q^{(m)}\v$ into the first equation of~\eqref{mom} gives
\be
Q^{(m)}
=Q^{(m)}\left(A^{(m)}+\frac{1}{u}\v\u^T\right).
\ee
Since the sum of the columns of $Q^{(m)}$ is $1-\v$ we get that $A^{(m)}+\frac{1}{u}\v\u^T$
is an $m\times m$ row stochastic matrix, and therefore from the assumption of the induction we get
\be
A^{(m)}+\frac{1}{u}\v\u^T=I_m.
\ee
However, since $\u,\v,A^{(m)}\geq 0$ the components of $\u$ and $\v$ must satisfy
\be
u_iv_j=0\;\;\text{ for }\;\;i\neq j.
\ee
Since both $\u\neq 0$ and $\v\neq 0$, there must exist $i_0$ such that $u_{i_0}\neq 0$, $v_{i_0}\neq 0$,
and $u_j=v_j=0$ for all $j\neq i_0$. However, in this case,
\be
\q_{m+1}=\frac{1}{u}Q^{(m)}\v=\frac{v_{i_0}}{u_{i_0}}\q_{i_0}.
\ee
That is, $\q_{m+1}$ is a multiple of another column of $Q^{(m+1)}$ contrary to the assumption that $Q^{(m+1)}$ is given in its standard form. Therefore, we must have $u=0$, and as discussed above, this corresponds to $A^{(m+1)}=I_{m+1}$. This completes the proof of the lemma.

Now, using this in the relation~\eqref{hhh} gives  that both $RT=I_m$ and $TR=I_\ell$.
Hence, $m=\ell$ and $R=T^{-1}$. But since both $R$ and $T$ are row stochastic, we must have that they are permutation matrices. Finally, since $Q$ and $P$ are given in their standard form, the permutation matrices $R$ and $T$ must be the identity matrices. This completes the proof of antisymmetry \eqref{antisymmetry} of $\prec_c$.
\end{proof}

\subsection{Linear programming methods}\label{subsec1E}
Consider the condition given in~\eqref{mqpt} for conditional majorization, and let
\be
\bt\equiv\begin{pmatrix}
\bt_1 &\\
\bt_2 &\\
\vdots &\\
\bt_m &
\end{pmatrix}\in\R^{m\ell},
\ee
where $\bt_1,...,\bt_m$ are the columns of $T$. Let $\Gamma$ be the $(nm+\ell)\times \ell m$ real matrix
\be
\Gamma=\begin{pmatrix}
-LP & \; & \; & \;\\
\; & -LP & \; & \;\\
\; & \; & \ddots & \;\\
\; & \; & \; & -LP\\
I_\ell & I_\ell &\cdots & I_\ell
\end{pmatrix},
\ee
where $I_\ell$ is the $\ell\times\ell$ identity matrix.
Finally, let
\be
\b\equiv\begin{pmatrix}
-L\q_1 &\\
\vdots &\\
-L\q_m &\\
\e
\end{pmatrix}\in\R^{(nm+\ell)},\ee
where $\q_1,...,\q_m$ are the columns of $Q$, and $\e\equiv (1,...,1)^T\in\R^\ell$.
With these notations we have the following proposition.
\begin{proposition}
With the same notations as above, $Q\prec_{c}P$ if and only if there exists $0\leq\bt\in\R^{m\ell}$ such that
\be
\Gamma\bt\leq\b.
\ee
\end{proposition}
\begin{proof}
The proof follows from Lemma~\ref{lem3}, with the observation that if there exists a matrix $T$ with non-negative entries that satisfies~\eqref{mqpt} with $\sum_{k=1}^{m}\bt_k\leq \e$, then there also exists a matrix $T'$ with non-negative entries that satisfies \eqref{mqpt} and $\sum_{k=1}^{m}\bt{'}_{k}=\e$.
\end{proof}

\subsection{Characterization of conditional majorization in terms of monotones}\label{subsec1F}
Here we present a more general version of the Theorem 1 from the main text that introduces the sufficient conditions for classical-conditional majorization in terms of monotones, i.e., convex functions that preserve the partial order under conditional majorization.

\begin{theorem}\label{mainresult}
Let $Q$ and $P$ be $n\times m$ and $n\times \ell$ joint probability matrices in the standard form. Then, the following conditions are equivalent:
\begin{enumerate}
\item $Q\prec_c P$.
\item For all matrices $A=[\a_1,...,\a_m]\in\R^{n\times m}_{+}$,
\be\label{mra}
\sum_{r=1}^{m}q_{r}\Phi_A(\q^{|r})\leq \sum_{r'=1}^{\ell}p_{r'}\Phi_A(\p^{|r'}),
\ee
where $\Phi_A(\q^{|r})\equiv \max_k \; (\a_{k})^{\downarrow} \bullet (\q^{|r})^{\downarrow}$.
\item For all convex symmetric functions $\Phi$,
\be
\sum_{r=1}^{m}q_{r}\Phi(\q^{|r}) \leq \sum_{r'=1}^{\ell}p_{r'}\Phi(\p^{|r'}).
\ee
\end{enumerate}
\end{theorem}
\begin{remark}
The matrix $A$ can be assumed to be in standard form.
\end{remark}
\begin{proof}
We first prove that $(1) \Longleftrightarrow (2)$. For any $P\in\R^{n\times\ell}_{+}$, denote by
\be
\label{markotop}
\cM(P, k):=\left\{Q'\in\R^{n\times k}_{+}\;:\;Q'\prec_c P\right\},
\ee
which we call the Markotop of $P$. Note that the Markotop of $P$ is a compact convex set.
Its vertices are the set of all matrices of the form $\cD[PT]$ (see \eqref{sqr_brackets1}-\eqref{sqr_brackets2}), where $\cD$ consists of $m$ permutation matrices and
$T$ is a matrix whose rows are elements of the standard basis in $\R^m$.

\begin{lemma}
\label{lem8}
Given an $n\times m$ matrix $Q$ and an $n\times \ell$ matrix $P$ then
\be
Q\prec_c P\;\;\iff\;\;\cM(Q, k)\subseteq\cM(P,k),\;\;\forall\;k\in\N.
\ee
\end{lemma}

\begin{proof}
Suppose $Q\prec_c P$ and let $A\in\cM(Q, k)$. Then, by definition, $A\prec_c Q$ and Theorem~\ref{thm1} gives $A\prec_c P$; that is, $A\in\cM(P, k)$. Conversely, $Q\in\cM(Q, m)\subseteq\cM(P,m)$ implies $Q\prec_c P$.
\end{proof}

Let now $\cS_{\cM(P,m)}: \R^{n\times m}\to\R$ be the support function of the Markotop of $P$  \eqref{markotop} defined by
\be
\cS_{\cM(P,m)}(A):=\max\left\{\tr(A^TQ)\;:\;Q\in\cM(P,m)\right\}
\ee
for any $A\in\R^{n\times m}$. Support functions of non-empty compact convex sets have the following property:
\be
\cM(Q,m)\subseteq\cM(P,m)\;\;\iff\;\;\cS_{\cM(Q,m)}\leq\cS_{\cM(P,m)}.
\ee
From Lemma \ref{lem8}, the support function provides a characterization of conditional majorization. We therefore
calculate
\be
\tr(A^TQ)\!=\!\sum_{k=1}^{n}\sum_{r=1}^{m}\!a_{kr}q_{kr}\!\!
=\!\!\sum_{r=1}^{m}\sum_{r'=1}^{\ell}t_{r'r}\!\sum_{x,k=1}^{n}D^{(r)}_{kx}a_{kr}p_{xr'}.
\ee
Since the set of doubly stochastic matrices is the convex hull of the permutation matrices we get
\be
\max_{D^{(r)}}\!\!\sum_{x,k=1}^{n}\!\!\!D^{(r)}_{kx}a_{kr}p_{xr'}\!=\!\max_{\pi}\!\sum_{x=1}^{n}\!\!a_{\pi(x)r}p_{xr'}\!
=\max_{\Pi}(\Pi\a_r)^T\p_{r'},
\ee
where the second maximum is over all permutations $\pi$, and the last maximum is over all $n\times n$ permutation matrices $\Pi$. Therefore,
\ben
\nonumber
\cS_{\cM(P,m)}(A)\!\!\! &= & \!\!\!\max_{\cD,T}\tr(A^T\cD[PT])\\
\nonumber
&=&\!\!\!\!\sum_{r'=1}^{\ell}\max_{r\leq m}\max_{\Pi}(\Pi\a_r)^T\!\p_{r'}\!=\!\sum_{r'=1}^{\ell}p_{r'}\Phi_A(\p^{|r'}),\\
\een
where
\be\label{hh}
\Phi_A(\p^{|r'})\equiv \max_{r\leq m}\max_{\Pi}(\Pi\a_r)^T\p^{|r'}=\max_{r\leq m}\;(\a_{r}^{\downarrow})^T\left(\p^{|r'}\right)^{\downarrow}.
\ee
Note that $\Phi_A$ is a convex, positively homogeneous (i.e., $\lambda\Phi_A(\p^{|r'})=\Phi_A(\lambda\p^{|r'})$ for $\lambda\geq 0$), and symmetric (under permutations of $\p^{|r'}$) function.

It remains to show that $(3) \Longleftrightarrow (1)$.
Indeed, suppose $Q\prec_{c} P$. Then, there exist family $\cD$ of $m$ $n\times n$ doubly stochastic matrices and an $\ell\times m$ stochastic matrix $T\equiv[\bt_1,...,\bt_m]$ such that
\be
Q=\cD[PT]=[D^{(1)}P\bt_1,...,D^{(m)}P\bt_m].
\ee
Denoting by $q_r\equiv\sum_{x=1}^{n}q_{xr}$ we therefore get that
\be
\label{eq45}
\q^{|r}=\frac{1}{q_r}D^{(r)}P\bt_r=\sum_{r'=1}^{\ell}\frac{t_{r'r}p_{r'}}{q_r}D^{(r)}\p^{|r'}.
\ee
Note that the RHS of Eq. \eqref{eq45} is a convex combination of $D^{(r)}\p^{|r'}$ since $q_r=\sum_{r'=1}^{\ell}t_{r'r}p_{r'}$. We therefore have
\ben
\nonumber
\sum_{r=1}^{m}q_{r}\Phi(\q^{|r})&=&\sum_{r=1}^{m}q_{r}\Phi\left(\sum_{r'=1}^{\ell}\frac{t_{r'r}p_{r'}}{q_r}D^{(r)}\p^{|r'}\right)\\
\nonumber
&\leq & \sum_{r=1}^{m}\sum_{r'=1}^{\ell}t_{r'r}p_{r'}\Phi\left(D^{(r)}\p^{|r'}\right)\\
\nonumber
&\leq & \sum_{r=1}^{m}\sum_{r'=1}^{\ell}t_{r'r}p_{r'}\Phi\left(\p^{|r'}\right)\\
&=&\sum_{r'=1}^{\ell}p_{r'}\Phi\left(\p^{|r'}\right),
\een
where the first inequality follows from the convexity of $\Phi$ and the second from its Schur convexity.\\
\end{proof}

Remembering that $P\in\R^{n\times\ell}_{+}$ and $Q\in\R^{n\times m}_{+}$, we now look at a few special cases of the matrix dimensions $n$, $\ell$ and $m$, wherein conditional majorization simplifies instantly to ordinary majorization:
\begin{enumerate}
\item $n=1$, with $\ell$ and $m$ arbitrary: In this case $P^{\downarrow}=Q^{\downarrow}$ always holds, and therefore $P\sim_cQ$ for any $(P,Q)$.
\item $m=1$, with $n$ and $\ell$ arbitrary: Here $Q=\q$ is a one-column matrix equivalent to a probability vector $\q$. Moreover, $T$ in Lemma~\ref{lem3} also have only one column, and since it is row-stochastic, $T=\e=(1,...,1)^T$. Denoting $\p\equiv P\e$, we have
\be
P\succ_cQ\;\;\iff\;\;\p\succ\q.
\ee
\item $\ell=1$, with $n$ and $m$ arbitrary: In this case $P=\p$ is a probability vector, and $T=(t_1,...,t_m)$ is a probability row vector.
We therefore get
\be\label{ell1}
P\succ_cQ\;\;\iff\;\;\p\succ\q^{|r}\;\;\forall\;r=1,...,m.
\ee
\item $\ell=2$, with $n$ and $m$ arbitrary. This case is more involved than the previous special cases. Details are presented below.
\end{enumerate}
Let us assume that $Q=Q^{\downarrow}$ and $P=P^{\downarrow}$ are 2-row joint probability matrices. For $\ell=2$, $Q\prec_{c}P$ if and only if there exists two $m$-dimensional probability vectors $\a$ and $\b$ such that $LQ\leq LPT$ \eqref{mqpt}, where 
\be\label{twot}
T=\begin{pmatrix}
a_1 & \cdots & a_m\\
b_1 & \cdots & b_m
\end{pmatrix}.
\ee
We denote by $q_r\equiv\sum_{x=1}^{n}q_{xr}$ for $r=1,...,m$, and similarly by $p_{r'}\equiv\sum_{x=1}^{n}p_{xr'}$ for $r'=1,2$. Since $p_1+p_2=1$ we simplify the notation and denote $p_1\equiv p$ and $p_2\equiv 1-p$. With these notations the Eq. \eqref{mqpt} is equivalent to 
\be\label{in1}
q_rL\q^{|r}\leq a_rpL\p^{|1}+b_r(1-p)L\p^{|2}\quad\forall\;r=1,...,m.
\ee
Note that the last component in the vector inequality above reads
\be
q_r\leq a_rp+b_r(1-p)\quad\forall\;r=1,...,m.
\ee
This is possible only if 
\be
q_r=a_rp+b_r(1-p)
\ee
for all $r=1,...,m$. Substituting this into~\eqref{in1} we  conclude that $Q\prec_{c}P$ if and only if there exists a probability vector $\a$ satisfying the following conditions:\\
\begin{align}
\label{cond_a}
& a_r\leq\frac{q_r}{p},\\
\label{cond_az}
& a_rL\left(\p^{|1}-\p^{|2}\right)\geq\frac{q_r}{p}L\left(\q^{|r}-\p^{|2}\right).
\end{align}

\begin{theorem}
\label{theoremW}
Let $P$ and $Q$ be $n\times 2$ and $n\times m$ probability matrices given in their standard form. Define
\be
\label{mi_nim}
\mu_k\equiv\sum_{x=1}^{k}\left(p_{x|1}-p_{x|2}\right)\;\;\text{and }\;\nu_{k}^{(r)}\equiv\sum_{x=1}^{k}\left(q_{x|r}-p_{x|2}\right).
\ee
Denote by $\cI^{+}$, $\cI^{0}$, and $\cI^{-}$ the set of indices $\{k\}$ for which
$\mu_k$ is positive, zero, and negative, respectively. Furthermore, define $p\equiv p_1=\sum_xp_{x1}$, so that $1-p=p_2$. Also define:
\be\label{alphabet}
\alpha_r\!\equiv\!\!\frac{q_r}{p}\max\left\{\!0,\max_{k\in\cI^{+}}\frac{\nu_{k}^{(r)}}{\mu_k}\!\!\right\}; \beta_r\!\equiv\!\!\frac{q_r}{p}\min\left\{\!1,\min_{k\in\cI^{-}}\frac{\nu_{k}^{(r)}}{\mu_k}\!\!\right\},
\ee
and through these,
\begin{align}
\label{cond_w1}
& W_0(P,Q)\equiv\!-\!\max_{r;k\in\cI^{0}}\{\nu_{k}^{(r)}\};\;\;W_{-}(P,Q)\equiv\left(\sum_{r=1}^{m}\beta_r\right)\!\!-\!\!1;\;\\
& W_{+}(P,Q)\equiv 1\!\!-\!\!\sum_{r=1}^{m}\alpha_r;\;\;W_1(P,Q)\equiv\min_{r\in\{1,...,m\}}\!\left(\beta_r\!-\!\alpha_r\right)\!.\label{cond_w2}
\end{align}
Then, $Q\prec_{c}P$ if and only if $W_0$, $W_1$, $W_+$, and $W_-$ are all non-negative.
\end{theorem}
\begin{proof}
To prove the theorem we need to show the equivalence between conditions \eqref{cond_a}-\eqref{cond_az} $\Leftrightarrow$ \eqref{cond_w1}-\eqref{cond_w2}. To prove the implication "$\Rightarrow$" we assume~\eqref{cond_a}-\eqref{cond_az} and express the condition~\eqref{cond_az} using notation introduced in~\eqref{mi_nim} in the following form
\be\label{inequal_az1}
a_r\mu_k\geq\frac{q_r}{p}\nu_{k}^{(r)}\quad\forall\;k, \forall r. 
\ee
Next, we consider $4$ cases.
\begin{enumerate}[(i)]
\item $k\in\cI^{0}$, i.e., $\mu_{k}=0$.

Then for all $r=1,...,m$,
\be
\frac{q_r}{p}\nu_{k}^{(r)}\leq 0 \;\forall\;(k\in\cI^{0}) \Rightarrow \frac{q_r}{p}\max_{k\in\cI^{0}}\nu_{k}^{(r)}\leq 0.
\ee
Since $\frac{q_r}{p}$ is positive for all $r$, we get that $W_0(P,Q)\equiv-\max_{k\in\cI^{0}}\{\nu_{k}^{(r)}\}$ is non-negative.
\item $k\in\cI^{+}$, i.e., $\mu_k>0$.

Then for all $r=1,...,m$
\be
a_r-\frac{q_r}{p}\frac{\nu_{k}^{(r)}}{\mu_k}\geq 0\quad\forall k\in\cI^{+}.
\ee
But since $a_r\ge0$ must also hold, we have
\be\label{inequal_az2}
a_r-\alpha_r\ge0.
\ee
Summing over $r$ and requiring $\sum_{r=1}^{m}a_r=1$, we see that $W_{+}\equiv 1-\sum_{r=1}^{m}\alpha_r$ must be non-negative.
\item $k\in\cI^{-}, i.e., \mu_k<0$.

Then for all $r=1,...,m$
\be
\frac{q_r}{p}\frac{\nu_{k}^{(r)}}{\mu_k}-a_r\geq 0\quad\forall k\in\cI^{-}.
\ee
But since $b_r\ge0$, \eqref{cond_a} implies that $a_r\le q_r/p$ must hold. Therefore,
\be\label{inequal_az3}
\beta_r-a_r\ge0.
\ee
Again, summing over $r$ and requiring $\sum_{r=1}^{m}a_r=1$ entails that $W_{-}\equiv\sum_{r=1}^{m}\beta_r-1$ be non-negative.
\item Combining~\eqref{inequal_az2} with~\eqref{inequal_az3} for all $r$, we obtain the condition $W_1\equiv\min_{r\in\{1,...,m\}}\left(\beta_r-\alpha_r\right)\ge0$.
\end{enumerate}
That completes the first part of the proof. 

To prove the other direction, define the matrices
\begin{align}
T_{\min}&\equiv
\begin{pmatrix}
\alpha_1&\alpha_2&\cdots&\alpha_m\\
\frac{q_1-p\alpha_1}{1-p}&\frac{q_2-p\alpha_2}{1-p}&\cdots&\frac{q_m-p\alpha_m}{1-p}
\end{pmatrix};\nonumber\\
T_{\max}&\equiv
\begin{pmatrix}
\beta_1&\beta_2&\cdots&\beta_m\\
\frac{q_1-p\beta_1}{1-p}&\frac{q_2-p\beta_2}{1-p}&\cdots&\frac{q_m-p\beta_m}{1-p}
\end{pmatrix}.
\end{align}
These matrices satisfy all the properties desired of the matrix $T$ in \eqref{twot}, except possibly row stochasticity, which they might fail by virtue of the rows not adding up to 1 (note that they don't fail in the nonnegativity condition on the entries). Now, the two row sums are continuous functions of the matrix elements. Moreover, if the first row adds up to 1, so does the second row. By the premise, $\sum_r\alpha_r=1-W_+\le1$ and $\sum_r\beta_r=W_-+1\ge1$, and $\beta_r\ge\alpha_r$ for all $r$. The intermediate value theorem ensures the existence of a row-stochastic $T$ with the desired properties.
\end{proof}

The conditions in the Theorem \ref{theoremW} can be simplified in special cases:
\begin{corollary}\label{gorol}
If  $\p^{|2}\prec\p^{|1}$ and also $\q^{|r}\prec \p^{|1}$ for all $r=1,...,m$ then $Q\prec_{c}P$ if and only if
$W_+(P,Q)\geq 0$. Furthermore, $W_+(P,Q)\geq 0$ if and only if
\be
\label{cor11a}
p\geq\sum_{r=1}^{m}q_r\max\{0,h_r\},\;\;\textrm{where }h_r\equiv\max_{k\in\cI^{+}}\left\{\frac{\nu_{k}^{(r)}}{\mu_k}\right\}.
\ee
Moreover, if $\p^{|1}=\e_1\equiv (1,0,...,0)^T$, then
\be
\label{cor11b}
h_r=\max_{k}\left\{\frac{\sum_{x=1}^{k}q_{x|r}-\pi_k}{1-\pi_k}\right\},\;\;\;\;\textrm{with}\;\;
\pi_k\equiv \sum_{x=1}^{k}p_{x|2}.
\ee
\end{corollary}
\begin{proof}
If $\p^{|2}\prec\p^{|1}$ then $\cI^{-}=\emptyset$. If in addition $\q^{|r}\prec \p^{|1}$ then $\mu_k\geq\nu_{k}^{(r)}$. In this case, we always have $W_0\geq 0$ and $W_1\geq 0$. Then, from Theorem \ref{theoremW}, $Q\prec_{c}P$ if and only if $W_+(P,Q)\geq 0$. Eq. \eqref{cor11a} and \eqref{cor11b} follow from direct calculations, with $\sum_{x=1}^{k}p_{x|1}=1$.
\end{proof}
Using this corollary, for any given arbitrary $n$-dimensional distribution $\bs\omega$, we can construct a ``minimal $P$" (that we call $\Omega$) that conditionally majorizes a given $Q$, with the requirement that $\bs\omega$ be one of the conditionals of this $\Omega$:
\begin{corollary}\label{mh1}
For an arbitrary $n$-dimensional distribution $\bs\omega=(\omega_1,\omega_2,...,\omega_n,0,...,0)^T$, define
$$R_{\bs\omega}\equiv\{r:\q^{|r}\prec_c\bs\omega\}.$$
If
\be
\sum_{r\in R_{\bs\omega}}q_r \geq 1-\alpha,
\label{eq:probz}
\ee
then $Q \prec_c \Omega$, where
\be
\Omega = 
\left[\begin{array}{ccccc} 
\alpha & 0 &  & \ldots& 0\\
(1-\alpha) \omega_1 & (1-\alpha) \omega_2  & \ldots & (1-\alpha) \omega_n 
\end{array}
\right]^T.
\ee
\end{corollary}
In the next section we construct a universal CUR by providing a bound in such a form as in the Corollary \ref{mh1}.

\subsection{Conditional majorization vs. ordinary majorization}\label{subsec1G}
In this subsection, we show the relation between the notion of conditional majorization, introduced in our work, in respect to the notion of majorization operation, developed in the literature \cite{Dahl1999}. In the case where the reference system is nontrivial, the two notions are completely unrelated, meaning that none is implied by the other. However, if we consider the reference system to be null, then conditional majorization simplifies to ordinary majorization (see points 1--4 in the subsection \ref{subsec1F}).

Below, we show two examples that illustrate the independence of conditional and ordinary majorization in a case of a bipartite scenario, where the nontrivial reference system belongs to Bob.

\textit{Example 1.} Let us compare two different states of Alice and Bob:
\ben
\label{state1}
\sigma &=& |0\>\<0|_A \ot \frac1n I_B,\\
\label{state2}
\tilde{\sigma} &=& \left( \frac{1}{10} |0\>\<0|+  \frac{9}{10} |1\>\<1| \right) _A \ot |0\>\<0|_B,
\een
using both conditional and ordinary majorization. For those two states, the conditional majorization implies that the state $\tilde{\sigma}$ \eqref{state2} is more conditionally uncertain than the state $\sigma$ \eqref{state1}, and the conditional majorization relation reads as $\tilde{\sigma} \prec_c \sigma$. This is due to the fact that having deterministic probability distribution $|0\>\<0| \sim (1,0,...,0)^T$, Alice can always create any other probability distribution, in particular: $\frac{1}{10} |0\>\<0|+  \frac{9}{10} |1\>\<1| \sim ( \frac{9}{10}, \frac{1}{10}, 0, ...0)^T$. However, the ordinary majorization implies the opposite, i.e., $\sigma \prec \tilde{\sigma}$, since  $\sigma \sim (\frac12, \frac12, 0,...,0)^T$ is more random (uncertain) than $(\frac{9}{10}, \frac{1}{10}, 0, ...0)^T$.

\textit{Example 2.} Now, let us compare another two states of a bipartite system, where the memory is held by Bob:
\ben
\label{state3}
\gamma &=&\left( \frac{1}{10} |0\>\<0|+  \frac{9}{10} |1\>\<1|\right)_A \ot \frac1n I_B,\\
\label{state4}
\tilde{\gamma} &=& |0\>\<0|_A \ot |0\>\<0|_B.
\een
Following the same arguments as in the Example 1, we obtain that in this case both ordinary and conditional majorization imply the same, i.e., that the state $\gamma$ \eqref{state3} is more uncertain and conditionally uncertain than  $\tilde{\gamma}$, where the respective majorization relations read as: $\gamma \prec \tilde{\gamma}$ and $\gamma \prec_{c} \tilde{\gamma}$.  

Based, on the above two examples, one sees that conditional majorization and ordinary majorization should be treated as separate independent notions, if the reference system is nontrivial.

\section{CUR derivations}\label{sec2}

\subsection{Derivation of a monotone-based CUR}\label{subsec2A}
In the derivation of the monotone-based CUR, we use the majorization UR stated in the main text in Eq. (2):
\be
\label{SMmajUR}
\q^{A_1} \ot \q^{A_2} \prec \bs\omega,
\ee
where $\q^{A_1}$ ($\q^{A_2}$) denotes the probability distribution over the outcomes of measurement of an observable $A_1$ ($A_2$) applied to an arbitrary quantum state $\rho$.  The bound $\bs\omega$ is given in the form:
\be
\label{omega}
\bs\omega =(\eta_1,\eta_2-\eta_1,...,\eta_{n}-\eta_{n-1},0,...0)^T,
\ee
where
\be
\label{eta_l}
\eta_l=\max_{\rho}\max_{\cI_l} \sum_{(a_1,a_2) \in \cI_l} q_{a_1} q_{a_2}.
\ee
In \eqref{eta_l}, $\cI_l \subset [n ] \times [n]$ denotes a subset of $l$ individual pair of indices $(a_1,a_2)$, where $[n]$ constitutes a set of natural numbers $1,...,n$. The quantity \eqref{eta_l} was derived explicitly in \cite{Friedland2013, Puchaa2013} for $l=1$: $\eta_1=\frac14(1+c)^2$ with $c=\max_{a_1,a_2} |\<a_1|a_2\>|$, and for $l=2$: $\frac14(1+c')^2$ with $c'=\max\sqrt{ |\<a_1|a_2\>|^2 +  |\<a'_1|a'_2\>|^2}$. For the remaining $\eta_l's$ for $l=\{3,...,n\}$ solely the upper bounds were provided.

Next, by applying to both sides of Eq. \eqref{SMmajUR} a convex function $\Phi$ that constitutes a proper measure of uncertainty, we obtain the UR: 
\be
\label{SMmaj_UR_Phi}
\Phi(\q^{A_1}\ot \q^{A_2}) \leq \gamma,
\ee
where $\gamma=\Phi(\bs\omega)$ depends on observables $A_1$ and $A_2$.

Let us now consider a bipartite scenario, where parties initially share two copies of a pure state $\Psi$, and we allow classical memory. Bob measures the same observable $R$ on his two copies, obtains result $r$ and prepares at Alice's site a state $\Psi_r$.  Alice performs measurement of observable $A_1$ on one of her copies, and on the other a measurement of $A_2$. She obtains the corresponding outcomes $a_1$ and $a_2$, which from now on we denote as a single variable $x \equiv (a_1,a_2)$. For each of the states $\Psi_r$, the UR \eqref{SMmaj_UR_Phi} still holds:
$\Phi(\q^{A_1}(\Psi_r)\ot \q^{A_2}(\Psi_r))=\Phi(\q^{A_1|r}\ot \q^{A_2|r}) \leq \gamma$. Next, by averaging over outcomes we get the corresponding CUR:
\be
\label{SMCURgeneral}
\sum_r q_r \Phi(\q^{A_1|r}\ot \q^{A_2|r}) \leq \gamma.
\ee
Let us rewrite the obtained CUR \eqref{SMCURgeneral} using a simplified notation:
\be
\label{SMmaj_CUR_Phi}
\sum_r q_r \Phi(\q^{|r}) \leq \gamma.
\ee
Now, we choose a particular monotone and derive a CUR based on it. Let us take a function $\Phi(\q^{|r})$ in Eq. \eqref{SMmaj_CUR_Phi} to be a function $\Phi_A(\q^{|r})\equiv \max_k \; (\a_{k})^{\downarrow} \bullet (\q^{|r})^{\downarrow}$ introduced in Eq. (7) in the main text with $A= (\underbrace{1,\ldots, 1}_{l},0,\ldots 0)^T$. Therefore, in the LHS of Eq. \eqref{SMmaj_CUR_Phi} the $l$-signed monotone is given by the function $\Phi_l(\q^{|r})=\sum_{x=1}^l q^{\downarrow}_{x|r}$. Applying the chosen monotone on $\gamma_l=\Phi_l(\bs\omega)$ with $\bs\omega$ defined in Eq. \eqref{omega} results in the monotone-based CUR in the final form: 
\be
\label{monotone_based_CUR}
\sum_{r} q_r \sum_{x=1}^l q^{\downarrow}_{x|r} \leq \eta_l
\ee
with $\eta_l$ provided in \cite{Friedland2013, Puchaa2013}: for $l=\{1,2\}$ explicitly, and for $l \geq 3$ upper bounded.

\subsection{Derivation of a universal CUR}\label{subsec2B}
Let $Q=[q_{xr}]$ be a matrix of joint probability distributions $q_{xr}=q_r q_{x|r}$ in the standard form. Therefore, the elements of each column $r$ are ordered as: $q_{x_1r} \geq q_{x_2r} \geq \cdots \geq q_{x_nr}$. Now, let us take $\sum_{x=1}^l q^{\downarrow}_{x|r}$ and average it over probability of columns $q_r$: $\sum_{r} q_r \sum_{x=1}^l q^{\downarrow}_{x|r}$. From the monotone-based CUR \eqref{monotone_based_CUR} derived in the previous subsection \ref{subsec2A}, we know that this expression is upper-bounded by $\eta_l$. Consider the monotone-based CUR \eqref{monotone_based_CUR} for the special case of $l=1$, where $\eta \equiv \eta_{l=1}$ is known explicitly:
\be
\label{monotone_based_CUR_l_1}
\sum_r q_r \max_x q_{x|r} \leq \eta
\ee
with 
\be
\label{eta}
\eta= \frac14 (1+c)^2.
\ee
It will be used in proving Theorem 2 from the main text, which provides a universal CUR: $Q \prec_c \Omega$, where $\Omega$ constitutes a bound we are looking for. 

To this end, we choose $\Omega$ in the form of a two-column matrix, where the first column is proportional to deterministic probability distribution and the second column is proportional to a vector:
\be
\label{vector_eta}
\bs\omega=(\underbrace{\beta, \beta, \beta}_{\ell}, 1 - \ell \beta, 0\ldots 0)^T,
\ee
which results in  the bound:
\begin{align}
\label{Omega}
    \Omega =
     \begin{pmatrix}
     \alpha
     \begin{bmatrix}
          1 \\
          0\\
           \\
           \\
           \vdots \\
          \\
          \\
           0
         \end{bmatrix}\quad (1- \alpha)
          \begin{bmatrix}
          \begin{rcases}
  \beta \\
  \vdots \\
  \beta
\end{rcases}
\ell\\         
           1 - \ell \beta\\
           0\\
           \vdots \\
           0
          \end{bmatrix}
    \end{pmatrix},
  \end{align}
where $\alpha$ is the parameter that makes the transformation $\Omega \mapsto Q$ achievable, cf. Corollary \ref{mh1}. The state $\Omega$ \eqref{Omega} can be equivalently written as:
\ben
\label{Omega_state}
\nonumber
\Omega &=&\alpha |0\> \<0|_{A} \ot |0\> \<0|_{B} \\
&+& (1-\alpha) |\omega\> \<\omega|^{A} \ot |1\> \<1|^{B}.
\een
Now it suffices to find an appropriate $\alpha$.

First, we state the elementary application of Markov's inequality in the form of the following lemma:
\begin{lemma}
\label{mh2}
Let $Q\in\cC\cC^n$. Suppose that 
\be
\label{markovassump}
\sum_r q_r \max_x q_{x|r}\leq \eta.
\ee
For arbitrary $\beta>0$, define 
\be
\label{Rbeta}
R_\beta\equiv\{r:\max\limits_x q_{x|r} \leq \beta\}.
\ee
Then,
\be
\label{markov}
\sum_{R_\beta}q_r \geq 1-\frac{\eta}{\beta}.
\ee
\end{lemma}
\begin{proof}
The proof follows directly from Markov's inequality.
\end{proof}

Let us substitute in Eq. \eqref{markov}: $\sum_{R_\beta}q_r \equiv 1-\alpha $, which immediately gives:
\be
\label{SMalphabeta}
\alpha \beta \leq \eta.
\ee
Then from Eqs.  \eqref{monotone_based_CUR_l_1} - \eqref{eta} we obtain the following constraint on the parameter $\alpha$:
\be
\label{SMalphabetac}
\alpha \beta \leq \frac14 (1+c)^2.
\ee

Now, we divide the set of parameters $r$ (numbering the columns of matrix $\Omega$) into two subsets: the subset of those $r \in R_{\beta}$ for which $\max_x q_{x|r} \leq \beta$, and the rest, $r \in \tilde{R}_{\beta}$ about which we know nothing. 

The procedure that allows the transformation $\Omega \mapsto Q$ is then the following. Bob looks at his register (see Eq. \eqref{Omega_state}). When he obtains  $|0\rangle$, he picks $r$ from the set $\tilde{R}_{\beta}$ with probabilities $q_r$ and tells Alice which $r$ he obtained. Then Alice creates $q_{x|r}$ out of $ |0\> \<0|_{A}$. This is always possible since from a deterministic probability distribution $\e \equiv (1, 0,0,...,0)^T$ she can produce any other probability distribution. On the other hand, when Bob obtains $|1\>$, he picks $r \in R_{\beta}$ with probabilities $q_r$, and informs Alice about $r$. In this case Alice creates $q_{x|r}$ out of $|\omega\> \<\omega|_{A}$. This is also always possible because for $r \in R_{\beta}$, $q_{x|r} \prec \omega$ by definition of $R_{\beta}$ \eqref{Rbeta}.

\section{Comparison with existing conditional uncertainty relations}\label{sec3}
\subsection{The universal CUR vs. entropic CUR}\label{subsec3A}
In this section we compare our universal CUR introduced in the main text in Theorem 2:
\begin{align}
\label{Suniversal_CUR}
    Q \prec_c \Omega \equiv
     \begin{pmatrix}
     \alpha
     \begin{bmatrix}
          1 \\
          0\\
           \\
           \\
           \vdots \\
          \\
          \\
           0
         \end{bmatrix}\quad (1- \alpha)
          \begin{bmatrix}
          \begin{rcases}
  \beta \\
  \vdots \\
  \beta
\end{rcases}
\ell\\         
           1 - \ell \beta\\
           0\\
           \vdots \\
           0
          \end{bmatrix}
    \end{pmatrix}
  \end{align}
with the entropic CUR -- a classical memory-assisted version of Maassen-Uffink UR given in Eq. (3) in the main text:
\be
\label{SRenes_clas}
H(A_1|\cR)+H(A_2|\cR)	\ge\log_2\frac1c.
\ee 
In \eqref{Suniversal_CUR}, $\ell$ denotes the largest integer such that $\beta \ell \leq 1$, and parameters $\alpha, \beta <1$ satisfy
\be
\label{alphabeta}
\alpha \beta = \frac14 (1 +c)^2,
\ee
where 
\be
\label{c}
\frac{1}{\sqrt{d}} \leq c \leq 1
\ee
is determined by the incompatibility of the measurements, and $d$ is the dimension of the Hilbert space. Using \eqref{alphabeta} and \eqref{c} with $\alpha, \beta < 1$ for $d \rightarrow \infty$, we obtain that
\be
\label{alpha_14}
\frac14 < \alpha < 1.
\ee

In Fig. 1 in the main text we described how to compare different CURs. In particular, CUR1 (here given by Eq. \eqref{Suniversal_CUR}) and CUR2 (given by Eq. \eqref{SRenes_clas}) are said to be independent whenever CUR1 excludes some statistics allowed by CUR2, but also allows some statistics excluded by CUR2. Let us now show that this is the case first by finding the statistics allowed by the universal CUR \eqref{Suniversal_CUR} but excluded by the entropic CUR \eqref{SRenes_clas}. 

Assume that Alice owns the system in the pure state, then, for the product states of Alice and Bob, our universal CUR simplifies to the ordinary majorization UR (as in Eq. (2) in the main text) and the conditional entropic relation \eqref{SRenes_clas} reduces to Maassen-Uffink UR (as in Eq. (1) in the main text). Then, based on the work of Friedland et. al. \cite{Friedland2013}, we know that there exist statistics that are excluded by the latter relation but satisfy the former. 

It remains to find a counterexample, namely such $Q$ that satisfies the entropic CUR, but for which our universal CUR is violated $Q \nprec_c \Omega$, in other words that there exists no protocol providing transformation $\Omega \mapsto Q$ defined in Eq. (4) in the main text. To this end, let us choose the following matrix Q:
\begin{align}
\label{Q_p}
    Q=
        \begin{pmatrix}
         p^2
          \begin{bmatrix}
          1 \\
          0\\
          0\\
          0\\
           \vdots \\
           0
         \end{bmatrix} \; 2p(1-p)
          \begin{bmatrix}
           \frac1d \\
           \vdots\\ 
           \frac1d \\ 
           0\\         
           \vdots \\
           0
          \end{bmatrix}\; (1-p)^2
          \begin{bmatrix}
           \frac{1}{d^2} \\ 
           \frac{1}{d^2} \\ 
            \frac{1}{d^2} \\       
            \frac{1}{d^2} \\     
           \vdots \\
            \frac{1}{d^2}
          \end{bmatrix}
    \end{pmatrix},
  \end{align}
whose elements are given in the form $q_{xr}=q_r q_{a_1|r} q_{a_2|r}$. By looking on the first columns $\q_1$ and $\om_1$ of matrices $Q$ \eqref{Q_p} and $\Omega$ \eqref{Suniversal_CUR}, respectively, one can easily see that there is impossible to realize transformation $\om_1 \mapsto \q_1$ (and so $\Omega \mapsto Q$) if
\be
\label{condition_violation}
\alpha <p^2.
\ee
\begin{remark}
We cannot take $p=\frac12$ since then we would need $\alpha < \frac14$, which would not obey Eq. \eqref{alpha_14}. Hence, we can take $p=0.51$, and $\alpha$ such that $\frac14 < \alpha < p^2$.
\end{remark}
Now, we notice that from the one side we would like to choose $c$ small since then we could take $\alpha$ so close to $\frac14$ for our universal CUR to be violated (violation is implied by \eqref{condition_violation}). On the other hand, we would also like to choose $c$ large enough for entropic CUR \eqref{SRenes_clas} to be satisfied. Let us then take:
\be
\label{c_d_gamma}
c=\frac{1}{d^\gamma}
\ee
for which $c \rightarrow 0$. This implies that $\alpha \rightarrow \frac14$. For such $c$, the entropic CUR \eqref{SRenes_clas} for the state \eqref{Q_p} reads:
$(1-p)^2 2 \log d + 2p(1-p) \log d \geq 2 \gamma \log d$. This requires $(1-p)^2 + p(1-p) \geq \gamma$, and therefore
\be
\gamma \leq 1-p.
\ee
Eventually, the entropic CUR \eqref{SRenes_clas} is satisfied for $\gamma \in [0,1-p]$ independently of $d$.\\

Next, let us analyze the violation of our universal CUR \eqref{Suniversal_CUR} determined by the condition \eqref{condition_violation}. This condition can be reformulated with the use of Eq. \eqref{c_d_gamma} and the fact that $\alpha < \frac14 (1+c)^2$ (derived for $\beta < 1$ from Eq. \eqref{alphabeta}) in the following form:
\be
\label{d_gamma_p}
\frac{1}{d^{\gamma}} < 2 p -1,
\ee
where $0 \leq \gamma \leq 1-p$. Then, for inequality \eqref{d_gamma_p} our universal CUR \eqref{Suniversal_CUR} is violated for some $p$ (e.g., $p=0.51$), since the LHS of \eqref{d_gamma_p} goes to zero for large $d$.\\

\begin{figure}[t]
  \centering
  \includegraphics[width=0.47\textwidth]{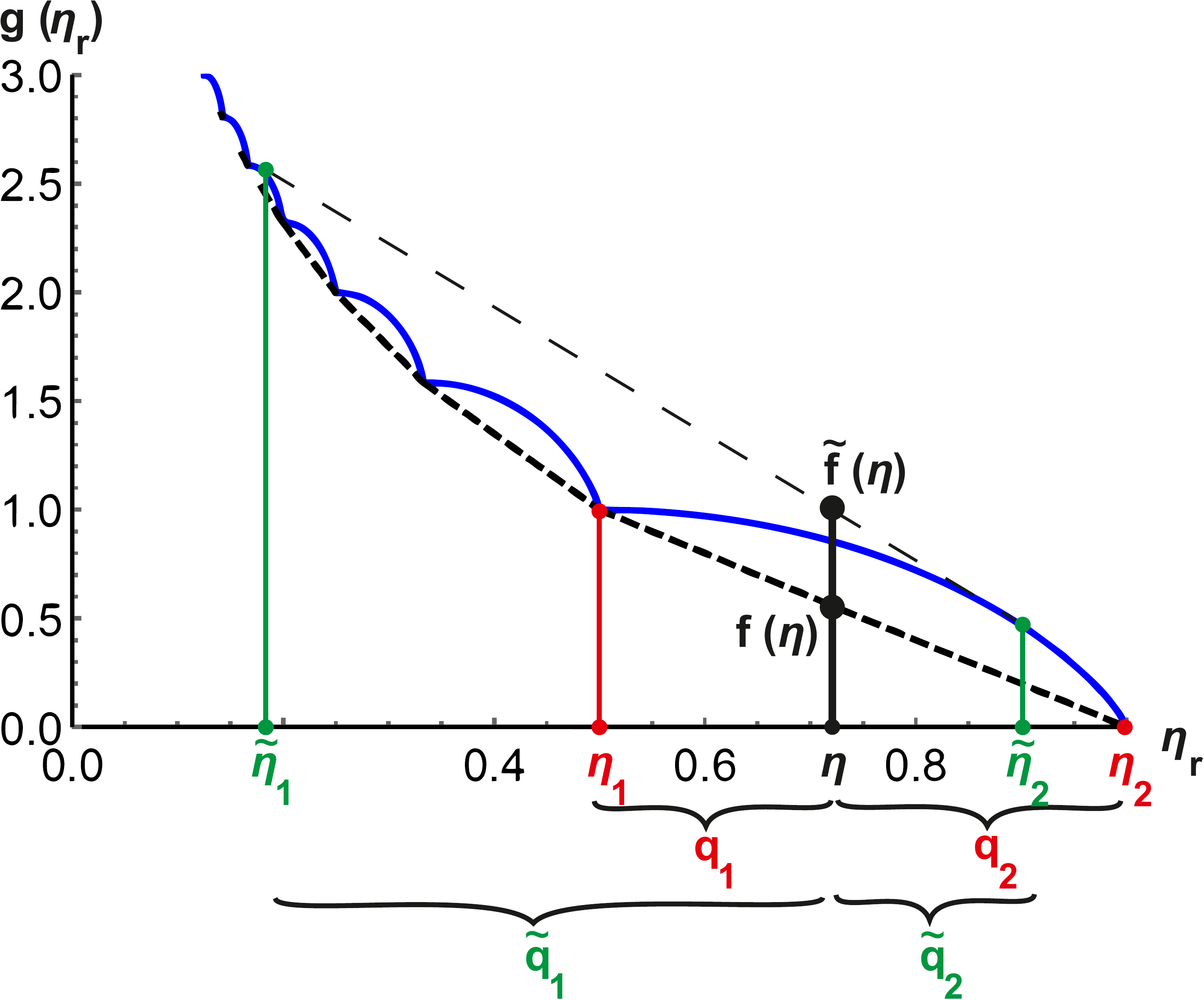}\\
  \caption{Graphical solution to the optimization problem \eqref{problem} under the constraint \eqref{constraint}. Blue curved polyline represents the function $g(\eta_r)$. Two exemplary convex combinations: $q_1 \eta_1 + q_2 \eta_2=\eta$ and $\tilde{q}_1 \tilde{\eta}_1 + \tilde{q}_2 \tilde{\eta}_2=\eta$ are highlighted by red and green, respectively, and let to different values of $f(\eta)$ and $\tilde{f}(\eta)$. The optimal solution of \eqref{problem} is given by black dashed polyline.}
  \label{function_eta}
\end{figure}

\subsection{Entropic CUR implied by the monotone-based CUR}\label{subsec3B}
Consider the monotone-based CUR for $l=1$ \eqref{monotone_based_CUR_l_1}, introduced in the main text as a special case of Eq. (13). In order to compare it with the entropic CUR \eqref{SRenes_clas}, we first need to find what entropic relation is implied by \eqref{monotone_based_CUR_l_1}. In particular, we construct a resulting entropic CUR given in a general form $H(X|\cR) \geq f(\eta)$ with $H(X|\cR)=\sum_r q_r H(\q^{X|r})$ and
\be
\label{f_eta_first}
f(\eta)=\min_{r} \sum_r q_r H(\q^{X|r}).
\ee

Let us then choose a probability vector that minimizes conditional entropy: 
\be
\label{vector_eta_prim}
\q^{X|r}=(\underbrace{\eta_r, \eta_r, \eta_r}_{\ell}, 1 - \ell \eta_r, 0\ldots 0)^T,
\ee
where
\be
\label{eta_r}
\eta_r=\max_x q_{x|r}
\ee
is the highest probability from the $r$-signed distribution for each $r$. Since \eqref{vector_eta_prim} is the probability vector in the standard form, we have that
\be
\label{eta_r_range}
\frac{1}{1+\ell} \leq \eta_{r} \leq \frac{1}{\ell}.
\ee

Eventually, the problem simplifies to finding a minimum of the following function:
\be
\label{problem}
f(\eta)= \min_r \sum_{r} q_{r} g(\eta_r),
\ee
under the constraint
\be
\label{constraint}
\min_r \sum_{r} q_{r} \eta_{r} = \eta,
\ee
where
\be
\label{eta_range}
\frac1m \leq \eta \leq 1
\ee
and $m$ denotes the number of (measurement) outcomes. In Eq. \eqref{problem} the function $g(\eta_r)$ reads as:
\be
\label{f_eta_second}
g(\eta_r)= H_b(\eta_{r} \ell) + \eta_{r} \ell \log \ell, 
\ee
and $H_b(p)=-p \log(p) -(1-p) \log(1-p)$ is a binary entropy. When stating the constraint \eqref{constraint}, we used \eqref{monotone_based_CUR_l_1} and \eqref{eta_r} with \eqref{eta_r_range} being satisfied. 
        
We solve the above problem graphically. Function $g(\eta_r)$ (see Eqs. \eqref{problem} and \eqref{f_eta_second}) is presented in Fig. \ref{function_eta} as a thick, blue line. To find the minimum \eqref{problem} under the constraint \eqref{constraint}, we need to find an optimal choice of free parameters $q_r$ and $\eta_r$. We illustrate it in Fig. \ref{function_eta}. Let us choose $\eta$ satisfying \eqref{eta_range}, and consider two exemplary convex combinations \eqref{constraint}: $q_1 \eta_1 + q_2 \eta_2=\eta$ (marked in red in Fig. \ref{function_eta}) and  $\tilde{q}_1 \tilde{\eta}_1 + \tilde{q}_2 \tilde{\eta}_2=\eta$ (marked in green). Note that $f(\eta) < \tilde{f}(\eta)$ and therefore the first choice is better. By considering any other possible convex combinations (with more elements) one can easily see that this is actually the optimal choice and that the dashed black polyline constitutes the desired minimum $f(\eta)$ \eqref{problem}. 

The last step is to represent the obtained result in terms of parameter $c$ according to Eq. \eqref{eta}. Eventually, we obtain the thick solid line from Fig. 2 presented in the main text.

\end{document}